\documentclass[11pt]{article}

\usepackage[margin=1in]{geometry}
\usepackage{amsmath,amssymb,amsthm}
\usepackage{mathtools}
\usepackage{microtype}
\usepackage{booktabs}
\usepackage{array}
\usepackage{etoolbox}
\usepackage[dvipsnames]{xcolor}
\usepackage{comment}
\usepackage{enumitem}

\definecolor{linkblue}{rgb}{0.0,0.0,0.55}
\definecolor{citegreen}{rgb}{0.0,0.35,0.0}
\usepackage[colorlinks=true,linkcolor=linkblue,citecolor=citegreen,urlcolor=linkblue]{hyperref}

\newtheorem{theorem}{Theorem}[section]
\newtheorem{lemma}[theorem]{Lemma}
\newtheorem{proposition}[theorem]{Proposition}
\newtheorem{corollary}[theorem]{Corollary}

\theoremstyle{definition}
\newtheorem{definition}[theorem]{Definition}
\newtheorem{remark}[theorem]{Remark}

\newcommand{\E}{\mathbb{E}}
\newcommand{\F}{\mathbb{F}}

\newcommand{\N}{\mathbb{N}}

\newcommand{\eps}{\varepsilon}
\newcommand{\dis}{\operatorname{dis}}
\newcommand{\wt}{\operatorname{wt}}
\newcommand{\supp}{\operatorname{supp}}

\newcommand{\CSS}{\mathrm{CSS}}
\newcommand{\Enc}{\mathrm{Enc}}

\newcommand{\QLR}{\textnormal{-QLR}}

\newcommand{\C}{\mathbb{C}}

\newcommand{\FRS}{\mathsf{FRS}}
\newcommand{\FQRS}{\mathsf{FQRS}}
\newcommand{\ev}{\operatorname{ev}}

\newcommand{\Tr}{\operatorname{Tr}}

\newcommand{\ket}[1]{\lvert #1\rangle}

\newtoggle{showcomments}
\newtoggle{showedits}

\toggletrue{showcomments} % Controls \xmc and \fc
\toggletrue{showedits}    % Controls \xmt and \ft

\title{Quantum List Recovery and Decoding: Achievability and Limitations}
\author{
  Fernando Granha Jeronimo\thanks{Siebel School of Computing and Data Science,
    University of Illinois Urbana-Champaign.
    \texttt{\{granha, xm20\}@illinois.edu}} \and Xiaojuan Ma\footnotemark[1]
}
\date{}

\providecommand{\GRS}{\mathrm{GRS}}
\providecommand{\RS}{\mathrm{RS}}

\begin{document}
\maketitle

\begin{abstract}
Quantum list recovery (QLR) and quantum list decoding (QLD) seek
short lists of logically distinct Pauli corrections consistent with a
syndrome and prescribed error constraints.  For CSS codes, two issues
arise relative to the classical setting: treating the \(X\)- and
\(Z\)-sectors separately can multiply their output list sizes, while
distinct classical candidates can collapse after quotienting by
stabilizers.

In this paper, we study combinatorial upper and lower bounds for
balanced folded quantum Reed--Solomon (FQRS) codes and balanced random
CSS codes, both studied by Bergamaschi, Golowich, and
Gunn (STOC 24).  For a fixed quantum rate \(R\in(0,1)\), let
\(R_1=(1+R)/2\) be the common component rate.  As the radius
\(\rho=(1-R)/2-\gamma\) approaches the quantum Singleton bound, the optimal worst-case list sizes for these code families satisfy
\[
  L^\star_{\rm QLR}
  =
  \ell^{\Theta(R_1/\gamma)},
  \qquad
  L^\star_{\rm QLD}
  =
  \Theta\!\left(\frac{1-R}{\gamma}\right)
  \qquad (\gamma\downarrow0).
\]
For every fixed list size \(L\ge1\), we also obtain the asymptotically
exact QLD radius tradeoff
\[
  \rho_L^\star
  =
  \frac{L}{L+1}\frac{1-R}{2}.
\]
These conclusions hold deterministically for FQRS codes and with high
probability for balanced random CSS codes, and extend to the
average-radius setting under the corresponding parameter conditions.

For achievability, building on the entropy and remainder inequalities
of Brakensiek, Chen, Dhar, and Zhang (STOC 2026), we establish a
pairing lemma for joint \(X/Z\) candidate lists that preserves the
one-sector coefficient and thereby avoids a product loss in list size.
For the QLD converse, we prove a quantum generalized Singleton bound
based on the classical projection-and-patching argument
with distinctness measured modulo the stabilizer. For the QLR lower bounds, we adapt the 
folded Reed--Solomon bad-list construction of Chen and Zhang (STOC 2025) so that the candidates
remain distinct modulo the stabilizer, and show separately that
classical bad lists survive the stabilizer quotient with high
probability for random CSS codes.
\end{abstract}

\clearpage
\begingroup
\small
\tableofcontents
\endgroup
\clearpage

\section{Introduction}
A central problem in quantum error correction is to understand how much
adversarial noise can be tolerated at a given coding rate.  For a quantum code of rate \(R\), the quantum Singleton bound limits
the fraction of correctable erasures to \((1-R)/2\).  For exact quantum error correction, the Knill--Laflamme
conditions imply that the adversarial decoding radius is at most half the
distance, and therefore at most \((1-R)/4\)~\cite{KL97,BGG24}.

Approximate quantum error correction can behave differently.
\cite{CGS05} constructed approximate quantum codes that allow
exponentially small error in the recovered state and whose decoding
radius essentially matches the quantum Singleton bound.  Their
construction, however, has asymptotically vanishing rate and requires an
exponentially large alphabet.  Later, Bergamaschi, Golowich, and
Gunn~\cite{BGG24} constructed efficiently decodable approximate quantum
codes over constant-size alphabets with decoding radius approaching
\((1-R)/2\) for every \(R\in(0,1)\).

\cite{BGG24} developed a stabilizer-code formulation of quantum list
decoding by counting the number of logically distinct low-weight Pauli
errors consistent with a given syndrome.  They then introduced folded
quantum Reed--Solomon codes, which are efficiently list decodable with
decoding radius approaching the quantum Singleton bound.  They further
developed a quantum version of the expander-based distance amplification
and alphabet reduction techniques of~\cite{AEL95} to obtain
constant-alphabet quantum list-decodable codes.  They also
introduced a quantum notion of \emph{list recovery}, where one seeks a
small number of logically distinct operators consistent with short lists
of candidate single-qudit Pauli operators on most coordinates.  This
primitive is used in their concatenated and distance-amplified
constructions: inner list decoding produces coordinate-wise candidate
logical errors, which are then list recovered by the outer stabilizer
code.  
Together, quantum list decoding and list recovery enable error
correction beyond the unique-decoding radius and play a role in
constructions of approximate quantum error-correcting codes.  Motivated
by these roles, we study quantitative combinatorial bounds for both
quantum list decoding and quantum list recovery.

Classically, list decoding~\cite{Eli57,Woz58} and list recovery have
seen long lines of combinatorial and algorithmic results.

\noindent\textbf{List decoding.}
In their foundational works, Guruswami and
Sudan~\cite{Sud97,GS98} gave the first list-decoding algorithms for
Reed--Solomon codes up to the Johnson bound.  The later breakthrough
work of Guruswami and Rudra~\cite{GR06Capacity,GR08} showed that folded
Reed--Solomon codes can approach list-decoding capacity.  Since then,
a sequence of works has substantially improved the output list size for
folded Reed--Solomon codes, from
\((1/\varepsilon)^{O(1/\varepsilon)}\)
~\cite{KRSW18,KRSW23,Tam24}, to \(O(1/\varepsilon^2)\)
~\cite{Sri25}, and finally to the asymptotically optimal
\(O(1/\varepsilon)\) bound of Chen and Zhang~\cite{CZ25}.
Guruswami and Narayanan~\cite{GN14} also considered a stronger notion of \emph{average-radius list
decoding}, which
controls the average distance of a candidate list from the received
word.

Shangguan and Tamo~\cite{ST20}, and independently
Roth~\cite{Rot22}, considered a generalized Singleton bound for
combinatorial list decoding.  In its asymptotic form, if a
rate-\(R_\mathrm{c}\) code has output list size at most \(L\) at decoding radius
\(\rho\), then
\(
    \rho\le \frac{L}{L+1}(1-R_\mathrm{c}).
\)
For \(L=1\), this recovers the usual Singleton limitation on the
unique-decoding radius.  More generally, it gives a fine-grained
tradeoff between decoding radius and list size; in particular, at
radius \(1-R_\mathrm{c}-\eps\), the bound requires
\(
    L\ge \frac{1-R_\mathrm{c}-\eps}{\eps}.
\)

Considerable work has gone into approaching this tradeoff.
Shangguan and Tamo~\cite{ST20} gave an explicit Reed--Solomon
construction attaining the generalized Singleton bound for output list
size \(2\).  Subsequent works showed that random puncturings of
Reed--Solomon and related codes approach an \(\eps\)-relaxed
average-radius version of the generalized Singleton bound
~\cite{BGM23,GZ23,AGL24}, with the required alphabet size reduced from
exponential in the block length to polynomial and eventually linear.
\cite{BDG24} showed that the exact attainment of the generalized Singleton bound for Reed--Solomon
codes can require exponentially large alphabets. Therefore, an
\(\eps\)-relaxation is necessary for smaller alphabet
sizes. The work \cite{CZ25} later showed that
folded Reed--Solomon codes explicitly and algorithmically approach the
relaxed average-radius generalized Singleton bound.

\noindent\textbf{List recovery.}
For list recovery, both structured and random linear codes have been
studied extensively.  For folded Reed--Solomon codes~\cite{Kra03,GR06Capacity},
at radius \(1-R_\mathrm{c}-\eps\), a sequence of works reduced the output list
size from \((n/\eps^2)^{O(\log\ell/\eps^2)}\) in~\cite{GR06Capacity}, to
\((n\ell/\eps)^{O(\ell/\eps)}\) in~\cite{GW13}, and then to
\((\ell/\eps)^{O((1+\log\ell)/\eps)}\)
in~\cite{KRSW18,KRSW23,Tam24}.  Most recently, Brakensiek,
Chen, Dhar, and Zhang~\cite{BCDZ26BL} obtained the bound
\((\ell/(R_\mathrm{c}+\eps))^{O(R_\mathrm{c}/\eps)}\).

For random linear codes, the Zyablov--Pinsker argument~\cite{ZP81}
gives output list size \(q^{O(\ell/\eps)}\) over alphabets of size
\(q=\ell^{O(1/\eps)}\) at radius \(1-R_\mathrm{c}-\eps\).
Rudra and Wootters~\cite{RW18} improved this to
\(q^{O(\log^2\ell/\eps)}\), and Li and Shagrithaya~\cite{LS25}
obtained the alphabet-independent bound
\((\ell/\eps)^{O(\ell/\eps)}\).  \cite{BCDZ26BL} further obtained
\((\ell/(R_\mathrm{c}+\eps))^{O(R_\mathrm{c}/\eps)}\).

The work of~\cite{BCDZ26BL}
provides a unified route to these near-capacity bounds, combining
discrete Brascamp--Lieb inequalities and subspace-design structure
with reductions to random linear and random Reed--Solomon codes.
Their results apply to explicit folded Reed--Solomon and univariate
multiplicity codes, as well as random linear and random Reed--Solomon
codes.

% \noindent\textbf{List-recovery lower bounds.}
On the converse side, list-recovery lower bounds developed through a
sequence of increasingly general settings.  In the zero-error regime,
Guruswami and Rudra~\cite{GR06Limits} showed that many full-length
Reed--Solomon codes require \(R_\mathrm{c}\le 1/\ell\) in order to have
polynomial output list size.  A series of later works showed that
approaching list-recovery capacity can require output list size
\(\ell^{\Omega(1/\eps)}\) for several important code families:
random linear codes in the high-rate zero-error regime
~\cite{GLMRSW22}, random linear codes in general parameter settings
~\cite{LMS25}, and Reed--Solomon, folded Reed--Solomon, and
multiplicity codes in general parameter settings~\cite{CZ25}.
Li and Shagrithaya~\cite{LS25} subsequently established a general
lower bound for arbitrary linear codes: at radius \(1-R_\mathrm{c}-\eps\),
every sufficiently long rate-\(R_\mathrm{c}\) linear code requires output list
size
\(
    L=\ell^{\Omega(R_\mathrm{c}/\eps)}.
\)
Their proof constructs a product-type bad list from
\(\Theta(R_\mathrm{c}/\eps)\) independent directions, with \(\ell\) choices
along each direction, yielding \(\ell^{\Theta(R_\mathrm{c}/\eps)}\) global
candidates while using at most \(\ell\) symbols on each coordinate of
a sufficiently large agreement set.  

The problem of determining the optimal rate--radius tradeoff for list
recovery, or equivalently a tight generalized Singleton bound, has
remained open, with partial progress over recent years
~\cite{GST24,RYZ24,CZ25,LMS25,LS25}.  Very recently, Brakensiek, Chen,
Putterman, and Zhang~\cite{BCPZ26} established a tight generalized
Singleton bound for list recovery.  They showed that, in the large-alphabet regime, the asymptotically
optimal rate threshold for \((\rho,\ell,L)\)-list recovery is
\(R_\mathrm{c}=(L+1-\ell)/L-(L+1)\rho/L\).
For \(\ell=1\), this recovers the generalized Singleton bound for list
decoding.
% Thus the Brascamp--Lieb upper bound above matches the known
% linear-code lower bound at the exponent scale.

\subsection{Quantum CSS code: Two questions}
These classical results naturally raise the question of whether similarly
combinatorial bounds hold for quantum list decoding and list recovery.
The translation is not immediate, however, because the objects being counted
in the quantum setting are different.  Classical list decoding returns a
short list of candidate codewords, whereas for a stabilizer code the encoded
quantum state itself is not revealed.  Instead, measuring the stabilizer
generators produces an error \emph{syndrome}, and the decoder outputs a list
of possible logically distinct Pauli errors consistent with that syndrome.
Here ``logically distinct'' is essential.  Pauli errors that differ by a
stabilizer act identically on the code space and therefore correspond to
the same logical correction, whereas errors that differ by a nontrivial
logical Pauli must remain distinct.  Thus quantum list decoding counts
stabilizer-distinct error classes rather than ordinary distinct Pauli
operators.

Calderbank--Shor--Steane (CSS) codes~\cite{CS96,Ste96} make this
stabilizer-decoding problem especially transparent in classical coding
terms.  A CSS code is built from two classical linear codes
\(C_1,C_2\subseteq\F_q^n\) satisfying \(C_2^\perp\subseteq C_1\).
The code \(C_1\) governs the \(X\)-component of a Pauli error, while
\(C_2\) governs the \(Z\)-component.  For two errors with the same
syndrome, their difference lies in \(C_1\oplus C_2\), while two such
errors represent the same logical correction precisely when their
difference lies in \(C_2^\perp\oplus C_1^\perp\).  Hence the logically
distinct corrections are naturally indexed by
\[
    (C_1\oplus C_2)/(C_2^\perp\oplus C_1^\perp).
\]
Thus CSS codes provide a bridge from quantum list decoding and
recovery to classical linear coding.

However, although CSS codes are built from classical codes, the classical
results above do not automatically transfer to quantum CSS codes because
of two additional features: the two-sector structure of Pauli errors and
stabilizer equivalence.

\begin{enumerate}
    \item \textbf{Two-sector structure:}
\emph{Can the \(X\)- and \(Z\)-sectors be analyzed jointly and avoiding the product of the two
marginal list sizes?}

A natural way to use classical list-recovery results is to treat the two
Pauli sectors separately.  Each register \(i\) comes with a list
\[
    \mathcal E_i\subseteq \F_q^s\times\F_q^s,
    \qquad |\mathcal E_i|\le \ell,
\]
of candidate Pauli symbols \((x,z)\).  Projecting \(\mathcal E_i\) onto
its two coordinates gives an \(X\)-list and a \(Z\)-list, each still of
size at most \(\ell\).  If the two corresponding marginal quotient problems each have
output list size at most \(L\), one may then combine the two outputs
by taking their Cartesian product.  This gives the following
general bound from~\cite{BGG24}:

\begin{proposition}[Sectorwise product bound;
  cf.~{\cite[Claim~B.3]{BGG24}}]%
\footnote{Numbering of results and sections cited
from~\cite{BGG24,BCDZ26BL,BCDZ26Matroid} follows the full
arXiv versions throughout.}
Let \(C_1,C_2\subseteq(\F_q^s)^n\) be \(\F_q\)-linear with
\(C_2^\perp\subseteq C_1\).
If \(C_1/C_2^\perp\) and \(C_2/C_1^\perp\) are each
\((\rho,\ell,L)\)-list recoverable, then
\(\CSS(C_1,C_2)\) is \((\rho,\ell,L^2)\)\QLR.
\end{proposition}

However, projection replaces the original paired constraint by a potentially strictly larger Cartesian product, and this possible loss is already visible in a single coordinate. Suppose
\[
    \mathcal E_i=\{(0,0),(1,1)\}.
\]
The original quantum input permits only these two paired symbols.
After projection, however, both marginal lists are \(\{0,1\}\), and
treating them independently allows the four pairs
\[
    (0,0),\ (0,1),\ (1,0),\ (1,1).
\]

The issue is therefore not that the classical component bounds are weak:
treating the two sectors separately discards the pairing between the
\(X\)- and \(Z\)-components of each candidate and may enlarge the
allowed paired list to a Cartesian product.

\item \textbf{Stabilizer equivalence:} 
\emph{What lower bounds and converses hold once stabilizer equivalence is
taken into account?}
Classical list bounds count distinct words, whereas quantum decoding
counts distinct stabilizer cosets.  In the CSS description, two errors
whose difference lies in
\(C_2^\perp\oplus C_1^\perp\) are stabilizer-equivalent and hence
represent the same logical correction.

This distinction can be substantial even within a single sector, as the
following simple example shows.  Let \(f\in C_1\), let
\(0\ne w\in C_2^\perp\), and choose distinct
\(\lambda_1,\ldots,\lambda_\ell\in\F_q\).  Then
\[
    f+\lambda_1 w,\quad
    f+\lambda_2 w,\quad
    \ldots,\quad
    f+\lambda_\ell w
\]
are \(\ell\) distinct classical words.  For each coordinate \(i\), define
\[
    S_i:=\{\,f_i+\lambda_j w_i:j\in[\ell]\,\}.
\]
Then \(|S_i|\le\ell\), and every word \(f+\lambda_jw\) satisfies
\[
    (f+\lambda_jw)_i\in S_i
    \qquad\text{for every }i.
\]
Thus all \(\ell\) words belong to the same zero-error list-recovery
instance \((S_1,\ldots,S_n)\).

On the other hand, for any distinct \(j,k\in[\ell]\),
\[
    (f+\lambda_j w)-(f+\lambda_k w)
    =
    (\lambda_j-\lambda_k)w
    \in C_2^\perp.
\]
Hence all \(\ell\) words represent the same coset of
\(C_1/C_2^\perp\).  Thus an entire classical bad list can collapse to a
single candidate after quotienting by the stabilizer.

\end{enumerate}

\subsection{Our contributions}
We answer both of the above questions for balanced folded quantum
Reed--Solomon (FQRS) codes and balanced random CSS codes.  On the
achievability side, we establish a pairing lemma showing that the
two-sector \(X/Z\) structure incurs no product loss: the corresponding
one-sector classical list-size bounds carry over to the joint quantum
list.  On the converse side, we show that stabilizer equivalence does
not collapse the corresponding classical obstructions.  For quantum
list decoding, we establish a quantum generalized Singleton bound with
the same fine-grained radius--list-size tradeoff as in the classical
setting.  For quantum list recovery, we construct stabilizer-distinct
bad lists showing that the classical lower-bound exponent survives the
stabilizer quotient.

Throughout, we use \(R_{\rm c}\) for the rate of a standalone classical
code, reserving \(R\) for the quantum rate and \(R_1\) for the common
component rate of a balanced CSS code, where balanced means that the two
component codes have the same dimension, and hence the same rate.  We likewise use
\(\varepsilon\) for classical slack parameters and \(\gamma\) for the
slack from the quantum Singleton radius unless otherwise stated.

Let \(Q=\CSS(C_1,C_2)\) be a CSS code on \(n\) registers over
\(\F_q^s\), with \(C_2^\perp\subseteq C_1\), and write
\(k_t:=\dim_{\F_q}C_t\).  We call the CSS code \emph{balanced} when
\(k_1=k_2\).  Its quantum rate and common component rate are
\(
  R:=\frac{k_1+k_2-sn}{sn},
  ~
  R_1:=\frac{k_1}{sn}
      =\frac{k_2}{sn}
      =\frac{1+R}{2}.
\)
Thus the quantum Singleton radius satisfies
\(
  \frac{1-R}{2}=1-R_1.
\)
Accordingly, when a classical result at rate \(R_{\rm c}\) and radius
\(1-R_{\rm c}-\varepsilon\) is applied to either component of a
balanced CSS code, we make the identification
\(
  R_{\rm c}=R_1,
  ~
  \varepsilon=\gamma,
\)
so that
\(
  1-R_{\rm c}-\varepsilon
  =(1-R)/2-\gamma.
\)

For FQRS codes, as in~\cite{BGG24}, \(C_1\) is a folded
Reed--Solomon code and \(C_2\) is a folded generalized
Reed--Solomon multiplier twist.  For random CSS codes, \(C_1\) and \(C_2\) are each marginally uniform
random linear codes, subject to the CSS nesting condition.

\subsubsection*{\textbf{1. Joint achievability without a sector-product
loss (to Question 1)}}

We identify entropy and remainder certificates that allow the two CSS
sectors to be combined without a sector-product loss.  We establish the required certificates on the relevant candidate
sets for FQRS codes and, with high probability, for random CSS codes.  

\begin{definition}[Coordinate certificates]
\label{def:certificate}
Let \(C\subseteq\Sigma^n\) be nonempty and let \(0<\delta\le1\).
For a random word \(Y\) supported on \(C\), write \(Y_i\) for its
\(i\)-th coordinate random variable and
\(r(Y):=1-\max_y\Pr[Y=y]\) for the remainder.
We say that \(C\) satisfies the entropy coordinate certificate
\(\mathsf{Cert}^{H}(\delta)\), respectively the remainder coordinate
certificate \(\mathsf{Cert}^{r}(\delta)\), if every such \(Y\) satisfies
\[
  \sum\nolimits_{i=1}^n H(Y_i)\ge \delta n H(Y),
  \qquad\text{respectively}\qquad
  \sum\nolimits_{i=1}^n r(Y_i)\ge \delta n r(Y).
\]
\end{definition}

\begin{lemma}[Informal version of Lemma~\ref{lem:pairing}]
\label{lem:pairing-intro}
Let \(\mathcal A\subseteq\Sigma_X^n\times\Sigma_Z^n\) be a paired
family, let \(\mathcal A_X\) be its \(X\)-projection, and let
\(\mathcal A_Z(x):=\{z:(x,z)\in\mathcal A\}\) be the \(Z\)-fiber over
\(x\).  For either the entropy or remainder certificate, if
\(\mathcal A_X\) satisfies the certificate with coefficient
\(\delta_X\) and every \(Z\)-fiber satisfies it with coefficient
\(\delta_Z\), then \(\mathcal A\) satisfies the corresponding joint
certificate with coefficient \(\min\{\delta_X,\delta_Z\}\).  In
particular,
\[
  \sum\nolimits_{i=1}^n H(X_i,Z_i)
  \ge \min\{\delta_X,\delta_Z\}\,nH(X,Z),
  \qquad
  \sum\nolimits_{i=1}^n r(X_i,Z_i)
  \ge \min\{\delta_X,\delta_Z\}\,nr(X,Z).
\]
\end{lemma}

For our list applications, it suffices to require these inequalities
for distributions supported on at most \(B\) words; we denote the
corresponding bounded-support certificates by
\(\mathsf{Cert}^{H}_B(\delta)\) and
\(\mathsf{Cert}^{r}_B(\delta)\).  The same pairing statement holds for
these bounded-support versions.

Brakensiek, Chen, Dhar, and Zhang~\cite{BCDZ26BL} use a discrete
entropic Brascamp--Lieb inequality and develop a remainder analogue
to derive such coordinate inequalities from subspace-design
hypotheses and convert them into list-size bounds.  Here we instead use the inequalities
themselves as an \emph{interface}: once suitable one-sector
certificates are available, the pairing lemma preserves their
coefficient across the two sectors.  The entropy- and
remainder-to-list conversions of~\cite{BCDZ26BL} can then be applied
directly to the resulting joint certificates, yielding QLR and QLD
bounds without multiplying two separate sectorwise list-size bounds.

How the one-sector certificates are established is code-family
specific.  For FQRS codes, subspace-design structure together with the
entropic and remainder Brascamp--Lieb inequalities provides the
required certificates.  For random CSS codes, we establish the bounded-support
certificates directly: a certificate violation forces a structured
local collision profile, and for fixed \(B\), the probability that any
such bad profile occurs in a random linear component code is
exponentially small.  Once these one-sector inputs are established,
the same pairing and list-conversion framework applies to both
families and removes the sector-product loss.

\subsubsection*{\textbf{2. Quotient-aware generalized Singleton bounds
for QLD (to Question 2)}}

To understand how stabilizer equivalence affects QLD, we prove a
\textbf{quotient-aware generalized Singleton bound} for scalar codes
that works directly with distinct quotient classes.  For lower-bound
purposes, it is sufficient to work in one sector: a quotient-distinct
bad list in either marginal quotient can be lifted to a
stabilizer-distinct bad list for the full quantum code by fixing the
other sector.

Our converse adapts the classical projection-and-patching proof of the
generalized Singleton bound~\cite{ST20} with distinctness measured
modulo the stabilizer subcode.

\begin{enumerate}[
  label=(\roman*),
  leftmargin=*,
  itemsep=2pt,
  topsep=2pt
]

\item \textbf{\textit{A one-sector quotient-aware generalized Singleton
bound.}}
For any nested scalar linear codes \(W\subseteq C\subseteq\F_q^n\)
, let
\(r_C:=n-\dim C\).  
For every \(L\) with \(L+1\le |C/W|\), we directly
construct representatives of \(L+1\) distinct cosets of \(W\) inside a
common Hamming ball of radius
\[
  \frac1n
  \left\lceil
    \frac{L}{L+1}
    \left(
      r_C+\left\lceil\log_q(L+1)\right\rceil
    \right)
  \right\rceil .
\]
For a balanced rate-\(R\) scalar CSS code, the marginal redundancy is
\(r_X=r_Z=(1-R)n/2\), so applying the bound to either marginal quotient
gives
\[
  \rho_L
  =
  \frac{L}{L+1}\frac{1-R}{2}+O(1/n),
\]
and hence \(L=\Omega((1-R)/\gamma)\) at radius
\((1-R)/2-\gamma\).

\item \textbf{\textit{Exponential coset-list growth and a balance
barrier.}}
The same argument shows that once the decoding radius exceeds the
one-sector redundancy \(r_C/n\), the largest coset list grows
exponentially:
\[
  \max_g |\mathcal L_{C/W}(g,\rho)|
  \ge
  \min\!\left\{
    q^{\lfloor\rho n\rfloor-r_C},
    |C/W|
  \right\}.
\]
Applying this to the two marginal quotients of a CSS code gives a
structural consequence.  If a rate-\(R\) scalar CSS family has
subexponential QLD list size at radius
\(\rho=(1-R)/2-\gamma\), then
\[
  |r_X-r_Z|
  \le
  2\gamma n+o(n).
\]
Thus any scalar CSS family approaching the quantum Singleton radius with
subexponential list size must be approximately balanced.

\end{enumerate}
The scalar projection-and-patching argument chooses a coordinate set
\(T\) containing an information set for the stabilizer subcode \(W\).
After folding, however, coordinates must be selected in whole
registers, and a scalar information set may be scattered across too
many registers for the same argument to go through. 
For FQRS codes, the MDS structure of the underlying Reed--Solomon
stabilizer subcode allows the same argument to be carried out over
whole folded registers.

\subsubsection*{\textbf{3. Quotient-aware lower bounds for QLR
(to Question 2)}}

For QLR, we construct one-sector bad lists that preserve the
\(\ell^{\Theta(R_1/\gamma)}\) scale of classical linear-code
list-recovery lower bounds while accounting for the stabilizer quotient.

\begin{itemize}[leftmargin=3em,noitemsep,topsep=2pt]

\item \textbf{FQRS.}
We construct a \emph{diagonal product grid} inspired by the
hypercube-like list-recovery construction of ~\cite{CZ25, LS25}, but modified to account for the stabilizer quotient.
In the FQRS realization,
\(C_2^\perp\subseteq C_1\) corresponds to evaluations of polynomials
of degree strictly less than a cutoff \(d_0\).  Hence two candidates
are stabilizer equivalent precisely when their polynomial difference
has degree below \(d_0\).

We choose the candidate generators to have distinct degrees in
\([d_0,k_1)\).  Consequently, the highest-degree term in every nonzero
candidate difference does not cancel.  Hence every nonzero pairwise
difference also has degree in \([d_0,k_1)\), and its evaluation lies in
\(C_1\setminus C_2^\perp\). This gives
\(\ell^{\lceil R_1/\gamma\rceil}\) stabilizer-distinct candidates at
radius \((1-R)/2-\gamma\).

\item \textbf{Random CSS.}
Here quotient survival follows from a direct dimension count.
Let \(M:=\ell^{\lfloor R_1/\gamma\rfloor}\) and start from a classical
rate-\(R_1\) bad list of size \(M\).  Conditioning on the component
code \(C_1\), which has dimension \(R_1n\), the stabilizer subcode
\(C_2^\perp\) is a uniformly random \((1-R_1)n\)-dimensional subspace
of \(C_1\).  Hence any fixed nonzero pairwise difference belongs to
\(C_2^\perp\) with probability at most \(2q^{-Rn}\).  A union bound
over fewer than \(M^2\) pairwise differences shows that, with probability
at least \(1-2M^2q^{-Rn}\), all \(M\) candidates remain distinct modulo
the random stabilizer subcode.

\end{itemize}

\subsubsection*{\textbf{4. Matching achievability and lower bounds}}

Combining the joint achievability of Contribution~1 with the
quotient-aware converses and lower bounds of Contributions~2--3 gives
matching asymptotic list-size scales at radius
\((1-R)/2-\gamma\).  Write \(L^\star_{\rm QLR}\) and \(L^\star_{\rm QLD}\) for the
smallest valid worst-case output-list bounds at the stated parameters. For FQRS codes, and with high probability for
random CSS codes,
\[
  \log_\ell L^\star_{\mathrm{QLR}}
  =
  \Theta\!\left(\frac{R_1}{\gamma}\right)
  =
  \Theta\!\left(\frac{1+R}{\gamma}\right),
\]
for fixed \(R>0\) and \(\ell\ge2\), subject to the stated
family-specific parameter conditions.
The sector-product loss does not change the asymptotic
\(\Theta(R_1/\gamma)\) exponent scaling.  It instead squares the
one-sector list-size bound, which doubles the constant in the exponent.
The pairing lemma removes this square and preserves the one-sector
list-size bound for the joint CSS problem.

For QLD, the effect is stronger at the level of list-size scaling.
The fixed-list remainder conversion underlying~\cite{BCDZ26BL} gives
a one-sector list size of order \(1/\gamma\) near the Singleton radius,
whereas a sector-product argument would give order \(1/\gamma^2\).
The pairing lemma preserves the one-sector certificate coefficient and
therefore the \(1/\gamma\) list-size order.

Together with our quotient-aware generalized Singleton converse, this
gives an asymptotically tight fixed-list radius tradeoff:
\[
  \rho_L^\star
  =
  \frac{L}{L+1}\frac{1-R}{2}.
\]
Equivalently, at radius
\((1-R)/2-\gamma\), both families satisfy
\[
  L^\star_{\mathrm{QLD}}
  =
  \Theta\!\left(\frac{1-R}{\gamma}\right),
\]
with the random-CSS achievability statement holding with high
probability.  Table~\ref{tab:results} summarizes these bounds and
compares them with the prior bounds of~\cite{BGG24}.

\begin{table}[h]
\centering
\footnotesize
\setlength{\tabcolsep}{10pt}
\renewcommand{\arraystretch}{1.25}
\begin{tabular}{@{}lllll@{}}
\toprule
& \multicolumn{2}{c}{BGG24}
& \multicolumn{2}{c}{This work} \\
\cmidrule(lr){2-3}\cmidrule(l){4-5}
Task
& Upper & Lower & Upper & Lower \\
\midrule
FQRS QLR
& \(\bigl(O(n/\gamma^2)\bigr)^{O(\log\ell/\gamma)}\)
& ---
& \(\left(\dfrac{\ell}{R_1+\gamma/2}\right)^{O(R_1/\gamma)}\)
& \(\ell^{\Omega(R_1/\gamma)}\) \\

FQRS QLD
& \(n^{O(1/\gamma)}\)
& ---
& \(O((1-R)/\gamma)\)
& \(\Omega((1-R)/\gamma)\) \\

Random CSS QLR
& ---
& ---
& \(\left(\dfrac{\ell}{R_1+\gamma/2}\right)^{O(R_1/\gamma)}\)
  (w.h.p.)
& \(\ell^{\Omega(R_1/\gamma)}\) (w.h.p.) \\

Random CSS QLD
& \(2^{O(1/\gamma^2)}\) (w.h.p.)
& ---
& \(O((1-R)/\gamma)\) (w.h.p.)
& \(\Omega((1-R)/\gamma)\) \\
\bottomrule
\end{tabular}

\caption{List-size bounds at the near-Singleton radius
\((1-R)/2-\gamma\), for fixed rate \(R>0\), with input-list size
\(\ell\ge2\) for QLR, as \(\gamma\downarrow0\).
The BGG24 entries~\cite{BGG24} are prior achievability bounds; a dash
indicates that no corresponding bound of that type is given there.
Our joint certificate framework avoids the sector-product loss that
would arise from separately bounding and multiplying the two marginal
lists.  For both FQRS and random CSS, the QLR upper and lower bounds
match at the exponent scale, while the QLD upper and lower bounds match
at the list-size order.
Entries marked ``w.h.p.'' hold with probability \(1-o_n(1)\) over the
balanced random-CSS ensemble; unmarked entries are deterministic. All entries are subject to the corresponding alphabet, folding,
and block-length conditions; these conditions may differ between
the prior results and this work.}
\label{tab:results}
\end{table}

\subsection{Technical overview}

We highlight three ideas behind the results above.  The first explains how
the two Pauli sectors can be combined without a product loss.  The second
builds the stabilizer quotient directly into the generalized-Singleton
converse for QLD.  The third is the quotient-aware construction underlying
the FQRS QLR lower bound.

\begin{enumerate}
\item \textbf{Pairing the two sectors before counting lists.}
Let \((X,Z)\) be a random element of a paired candidate family
\(\mathcal A\), let \(\mathcal A_X\) be its \(X\)-projection, and let
\(\mathcal A_Z(x):=\{z:(x,z)\in\mathcal A\}\) be the \(Z\)-fiber over
\(x\).

\begin{itemize}[leftmargin=1.5em,itemsep=2pt,topsep=2pt]
% label=\textendash,
\item
For entropy, the mechanism is the chain rule:
\[
  \sum\nolimits_i H(X_i,Z_i)
  =
  \sum\nolimits_i H(X_i)+\sum\nolimits_i H(Z_i\mid X_i)
  \ge
  \sum\nolimits_i H(X_i)+\sum\nolimits_i H(Z_i\mid X).
\]
The \(X\)-certificate controls the first term.  Conditioned on \(X=x\),
the random variable \(Z\) is supported on 
\(\mathcal A_Z(x)\), so the \(Z\)-fiber certificate controls the second.
This gives
\[
  \sum\nolimits_i H(X_i,Z_i)
  \ge
  \min\{\delta_X,\delta_Z\}\,nH(X,Z).
\]

\item
For remainder \(r(Y):=1-\max_y\Pr[Y=y]\), there is no analogous
additive chain rule.  Instead, remainder admits a distance
interpretation:
\[
  \sum\nolimits_i r(Y_i)
  =
  \min\nolimits_y \mathbb E[d_H(Y,y)].
\]
Thus, for the paired word \((X,Z)\),
\(
  \sum\nolimits_i r(X_i,Z_i)
  =
  \min_{(u,v)}
  \mathbb E\!\left[d_H\bigl((X,Z),(u,v)\bigr)\right].
\)
This lets us reason directly in terms of paired Hamming disagreement.

Fix a paired center \((u,v)\) and partition the candidate distribution
into fibers according to its \(X\)-word.  For a fiber with \(X\)-word
\(x\), let \(E_x:=\{i:x_i\ne u_i\}\).  Coordinates in \(E_x\) have
already been counted as paired errors through the \(X\)-sector, so the
\(Z\)-sector only needs to account for disagreements on \(E_x^c:=[n]\setminus E_x\).
This prevents the two sectors from counting the disagreement of the same register
twice.

The probability mass has a corresponding fiber decomposition.  Choose a
maximum-mass candidate in each \(X\)-fiber, let \(a_x\) denote its mass,
set \(A:=\sum_x a_x\), and let \(a_\star:=\max_x a_x\).  Then
\[
  r(X,Z)
  =
  1-a_\star
  =
  \underbrace{A-a_\star}_{\text{among fibers}}
  +
  \underbrace{1-A}_{\text{within fibers}}.
\]
The \(X\)-certificate controls the first term through the chosen fiber
representatives, while the \(Z\)-fiber certificates control the second
through the residual mass inside each fiber.  Combining the two gives
\[
  \sum\nolimits_i r(X_i,Z_i)
  \ge
  \min\{\delta_X,\delta_Z\}\,n r(X,Z).
\]

\end{itemize}

\item \textbf{A quotient-aware generalized Singleton bound for QLD.}

The classical generalized-Singleton argument projects onto a suitably chosen set of coordinates so that, by pigeonhole, \(L+1\) codewords have
the same projection; it then patches the remaining coordinates to place
these codewords in a common Hamming ball.  In the quantum setting,
however, \(L+1\) distinct codewords are not enough: they must represent
\(L+1\) distinct logical classes.

Let \(W\subseteq C\subseteq\F_q^n\), where \(W\) is the one-sector
stabilizer subcode and \(C/W\) indexes the corresponding logical classes.
Choose one representative from each coset, forming a transversal
\(\mathcal U\subseteq C\), so that every codeword can be written uniquely
as \(u+w\) with \(u\in\mathcal U\) and \(w\in W\).

For a target list size \(L\), set
\(
  h:=\lceil\log_q(L+1)\rceil
\)
and choose
\(
  |T|=\dim C-h.
\)
Projection onto \(T\) has at most \(q^{|T|}\) possible values, whereas
\(|C|=q^{\dim C}\).  Hence some projection fiber contains at least
\(q^h\ge L+1\) codewords.

To prevent these distinct codewords from collapsing after stabilizer
quotienting, we additionally choose \(T\) to contain an information set
of \(W\).  Then \(\pi_T\) is injective on \(W\): for any fixed
\(u\in\mathcal U\), at most one word \(u+w\) from the coset \(u+W\)
can lie in a given projection fiber.  Therefore the \(L+1\) codewords
found by pigeonhole necessarily represent \(L+1\) distinct quotient
classes.  In the scalar case, \(L+1\le |C/W|\) implies
\(
  h\le \dim C-\dim W,
\)
and hence \(|T|=\dim C-h\ge\dim W\). Therefore, \(T\) can indeed be chosen to
contain an information set of \(W\).

The remaining patching step is the same as in the classical argument.
The \(L+1\) representatives already agree on \(T\); partitioning the
remaining coordinates among them produces a common center at radius
\[
  \frac1n
  \left\lceil
    \frac{L}{L+1}
    \left(
      n-\dim C+\left\lceil\log_q(L+1)\right\rceil
    \right)
  \right\rceil .
\]
Thus the generalized-Singleton construction survives passage from
distinct codewords to distinct stabilizer cosets.

For folded codes, a scalar information set may be scattered across
too many registers.  In the FQRS application, however,
\(W=C_2^\perp\) is obtained by folding an underlying Reed--Solomon
code.  By the MDS property, any \(\dim_{\F_q}W\) scalar evaluation
coordinates determine \(W\), and hence any
\(\lceil\dim_{\F_q}W/s\rceil\) whole registers suffice for
injectivity.  Thus the projection-and-patching argument extends to
the FQRS quotient.  We discuss this folded-register issue in detail in
Section~\ref{subsec:folded-fqrs-quotients}.

\item \textbf{The diagonal product grid for the FQRS QLR lower bound.}

The challenge is to construct a large bad list that simultaneously
satisfies the list-recovery condition and stabilizer distinctness.
We use the balanced FQRS realization following
\cite[Definition~3.8]{BGG24}.  Both \(C_1\) and \(C_2^\perp\) are FRS
codes on the same evaluation points, while
\(C_2:=(C_2^\perp)^\perp\) is a folded GRS code by duality.
Recall \(R\) is the quantum rate and \(R_1:=(1+R)/2\) is the common component rate of a balanced CSS code.  Let
\(k_1:=\dim_{\F_q}C_1=R_1sn\) and
\(d_0:=\dim_{\F_q}C_2^\perp=(1-R_1)sn\).
Writing \(\ev\) for the folded evaluation map, we have
\[
  C_2^\perp
  =
  \ev(\F_q[x]_{<d_0})
  \subseteq
  C_1
  =
  \ev(\F_q[x]_{<k_1}).
\]
Thus \(C_2^\perp\) consists of evaluations of the low-degree
polynomials, while the quotient \(C_1/C_2^\perp\) is represented by
the coefficient band of degrees \([d_0,k_1)\).  Accordingly, if
\(f_1,f_2\in\F_q[x]_{<k_1}\) are the message polynomials of
\(y_1=\ev(f_1),y_2=\ev(f_2)\in C_1\), then
\[
  y_1\not\equiv y_2 \pmod{C_2^\perp}
  \quad\Longleftrightarrow\quad
  \ev(f_1-f_2)\notin C_2^\perp
  \quad\Longleftrightarrow\quad
  \deg(f_1-f_2)\ge d_0 .
\]
Stabilizer distinctness can therefore be enforced through polynomial
degree.

For the classical FRS list-recovery lower bound, Chen and
Zhang~\cite{CZ25} introduced a hypercube-like construction that
produces \(\ell^b\) distinct classical candidates while each designated
folded coordinate sees at most \(\ell\) different symbols. We rephrase their construction in our notation and adapt it to our setting. At a high
level, one partitions a collection of registers into \(b\) pairwise
disjoint groups \(G_1,\ldots,G_b\), each of size \(r\), and defines
\(
  v_j(x)
  :=
  \prod_{h\ne j}
  \prod_{i\in G_h}
  \prod_{m=1}^s
  (x-\alpha_{i,m}),
\)
where \(\alpha_{i,m}\) are the scalar evaluation points in register
\(i\).  Hence \(v_j\) vanishes on every designated group except
\(G_j\).  If \(T\subseteq\F_q\) has size \(\ell\) and
\(f=\sum_{h=1}^b t_hv_h\) with \(t_h\in T\), then for
\(i\in G_j\) we have
\(\pi_i(f)=t_j\pi_i(v_j)\).  Thus the \(b\) independent coefficient
choices produce \(\ell^b\) global candidates, while each designated
register sees at most \(\ell\) folded symbols.

The quantum setting further requires these classically distinct
candidates to remain distinct modulo \(C_2^\perp\).  The generators
\(v_j\) above all have the same degree \((b-1)rs\), so their leading
terms may cancel in a difference of two candidates.  Consequently,
even if the individual generators have degree at least \(d_0\), the
degree of a nonzero candidate difference may drop below \(d_0\).

We resolve this by replacing \(v_j\) with
\[
  u_j(x)
  :=
  x^{j-1}v_j(x)
  =
  x^{j-1}
  \prod_{h\ne j}
  \prod_{i\in G_h}
  \prod_{m=1}^s
  (x-\alpha_{i,m}),
  \qquad j\in[b].
\]
Since all scalar evaluation points are nonzero, the factor
\(x^{j-1}\) does not change the diagonal vanishing pattern:
\(\pi_i(u_j)=0\) for \(i\in G_h\), \(h\ne j\), while
\(\pi_i(u_j)\ne0\) for \(i\in G_j\).  Its role is instead to separate
the generator degrees, giving
\(\deg u_1<\deg u_2<\cdots<\deg u_b\).  For our parameters we take
\(b=\lceil R_1/\gamma\rceil\) and choose \(r\) so that all these
degrees lie in the window \([d_0,k_1)\), while satisfying the
radius \(\rho=(1-R)/2-\gamma\) list-recovery agreement condition.   The highest-degree
generator in any nonzero candidate difference then cannot cancel. Therefore, every nonzero candidate difference has degree in
\([d_0,k_1)\), representing a valid codeword in \(C_1\) while
remaining outside \(C_2^\perp\).

The construction thus preserves the classical diagonal product
mechanism of~\cite{CZ25} while adding the degree separation needed for
survival under the stabilizer quotient, yielding
\(
  L^\star_{\rm QLR}
  \ge
  \ell^{\lceil R_1/\gamma\rceil}.
\)
The complete construction and parameter verification appear in
Section~\ref{subsec:fqrs-soft-grid}.

\end{enumerate}

\paragraph{Roadmap.}
Section~\ref{sec:prelim} develops the preliminaries, including the
syndrome--quotient formulation of QLR and QLD, the FQRS and random-CSS
families, and the classical subspace-design and Brascamp--Lieb inputs.
Section~\ref{sec:qlr-upper} introduces the entropy and remainder
certificates and proves the pairing lemma and its bounded-support
form.
Section~\ref{sec:quotient-tools} develops the quotient-aware generalized
Singleton converse, first for scalar quotients and then for folded
registers.
Sections~\ref{sec:fqrs-results} and~\ref{sec:random-css-results}
instantiate these tools for FQRS and random CSS codes, respectively,
combining achievability with quotient-aware lower bounds to obtain the
matching list-size scales.

\section{Preliminaries}
\label{sec:prelim}

Throughout, \(q=p^m\) is a prime power.  Unless stated otherwise, all
vector spaces and linear codes are over \(\F_q\), all dimensions are
\(\F_q\)-dimensions, and all duals are taken over \(\F_q\). For convenience, we occasionally restate prior results in notation
adapted to ours, or isolate consequences of their proofs or remarks
that are useful for our arguments.

\subsection{CSS codes, quantum list problems, and quotient reduction}
\label{sec:stab}
\label{sec:qlr}

\paragraph{Generalized Paulis.}
A \(q\)-ary qudit has Hilbert space \(\C^q\) with computational basis
\(\{\ket{x}\}_{x\in\F_q}\).  Let
\(\Tr_{\F_q/\F_p}(a):=a+a^p+\cdots+a^{p^{m-1}}\) and
\(\omega:=e^{2\pi i/p}\).
For \(a,b\in\F_q\), define
\(X(a)\ket{x}:=\ket{x+a}\) and
\(Z(b)\ket{x}:=\omega^{\Tr_{\F_q/\F_p}(bx)}\ket{x}\).
For \(a,b\in\F_q^N\), write
\(X(a):=\bigotimes_jX(a_j)\),
\(Z(b):=\bigotimes_jZ(b_j)\), and
\(E_{a,b}:=X(a)Z(b)\).
We call \((a,b)\in\F_q^N\oplus\F_q^N\) the \emph{Pauli label} of
\(E_{a,b}\).  The generalized Pauli group consists of the operators
\(\omega^cE_{a,b}\), with \(c\in\F_p\).

A direct computation gives
\(E_{a,b}E_{a',b'}
=\omega^{\Tr_{\F_q/\F_p}(a'\cdot b)}
E_{a+a',b+b'}\)~\cite{AK01,BGG24}.
Hence
\(E_{a,b}E_{a',b'}
=\omega^{\Tr_{\F_q/\F_p}(a'\cdot b-a\cdot b')}
E_{a',b'}E_{a,b}\).
Thus Pauli labels add under multiplication up to phase, and two Pauli
operators commute precisely when
\(\Tr_{\F_q/\F_p}(a'\cdot b-a\cdot b')=0\).

\paragraph{Stabilizer codes, normalizers, and syndromes.}
A \emph{stabilizer code} \(Q\) is the common \(+1\)-eigenspace of an
abelian Pauli subgroup \(\mathcal S(Q)\) containing no nontrivial
scalar operator~\cite{AK01,BGG24}.
Its \emph{normalizer} \(\mathcal N(Q)\) consists of the Pauli
operators commuting with every element of \(\mathcal S(Q)\), and
\(\mathcal N(Q)/\mathcal S(Q)\) describes the logical Pauli operators.

Two Pauli errors \(E,E'\) are \emph{stabilizer-equivalent} if
\(E^\dagger E'\in\mathcal S(Q)\), up to phase.
Stabilizer-equivalent errors have the same action on the code space up
to a global phase.

Fix stabilizer generators \(S_1,\ldots,S_r\).  The \emph{syndrome} of
a Pauli error \(E\) is
\(\operatorname{syn}_Q(E):=(\sigma_1,\ldots,\sigma_r)\in\F_p^r\),
where \(S_jE=\omega^{\sigma_j}ES_j\).
Syndromes are additive, and \(\operatorname{syn}_Q(E)=0\) exactly
when \(E\in\mathcal N(Q)\).  Consequently, two errors \(E,E'\) have
the same syndrome exactly when \(E^\dagger E'\in\mathcal N(Q)\), up
to phase.

\paragraph{CSS codes.}
Let \(C_1,C_2\subseteq(\F_q^s)^n\cong\F_q^{sn}\) be
\(\F_q\)-linear codes of dimensions \(k_1,k_2\), with duals taken
with respect to the standard \(\F_q\)-bilinear dot product on
\(\F_q^{sn}\), and suppose
\(C_2^\perp\subseteq C_1\) (equivalently
\(C_1^\perp\subseteq C_2\)).
For \(Q=\CSS(C_1,C_2)\), ignoring phases, the stabilizer and normalizer
label spaces are
\(\mathcal S_Q=C_2^\perp\oplus C_1^\perp\) and
\(\mathcal N_Q=C_1\oplus C_2\), where the two components are the
\(X\)- and \(Z\)-Pauli labels.  Hence
\[
  \mathcal N_Q/\mathcal S_Q
  \cong
  (C_1/C_2^\perp)\oplus(C_2/C_1^\perp).
\]
The code encodes \(k_1+k_2-sn\) \(q\)-dimensional logical qudits and
has quantum rate \(R:=(k_1+k_2-sn)/(sn)\).
When \(k_1=k_2\), we call the CSS code \emph{balanced}; then
\(R_1:=k_1/(sn)=k_2/(sn)=(1+R)/2\).

Writing \(a=(a_1,\ldots,a_n)\) and \(b=(b_1,\ldots,b_n)\), with
\(a_i,b_i\in\F_q^s\), define the block weight of \(F=E_{a,b}\) by
\(\wt(F):=|\{i\in[n]:(a_i,b_i)\ne(0,0)\}|\).
Quantum list recovery allows a short list of possible Pauli labels at
each register.  For \(i\in[n]\), let
\(\mathcal E_i\subseteq\F_q^s\oplus\F_q^s\) satisfy
\(|\mathcal E_i|\le\ell\).
For a paired label
\(c=(a,b)\in(\F_q^s\oplus\F_q^s)^n\), define
\(\dis(c,\mathcal E):=
|\{i\in[n]:(a_i,b_i)\notin\mathcal E_i\}|\).
For \(F=E_{a,b}\), we also write
\(\dis(F,\mathcal E):=\dis((a,b),\mathcal E)\).
We use the following quantum list-decoding and list-recovery
definitions.

\begin{definition}[Quantum list decoding;
{\cite[Definition~3.2]{BGG24}}]
\label{def:qld-intro}
A stabilizer code \(Q\) of block length \(n\) is
\((\rho,L)\)-\emph{quantum list decodable} (QLD) if, for every
possible syndrome \(\sigma\), there are at most \(L\) pairwise
stabilizer-distinct Pauli errors of weight at most \(\rho n\)
consistent with \(\sigma\).
\end{definition}

\begin{definition}[Quantum list recovery;
cf.~{\cite[Definition~B.2 and Claim~B.2]{BGG24}}]
\label{def:qlr-intro}
A stabilizer code \(Q\) of block length \(n\) is
\((\rho,\ell,L)\)-\emph{quantum list recoverable} (QLR) if, for every
syndrome \(\sigma\) and every collection of local Pauli lists
\(\mathcal E_1,\ldots,\mathcal E_n\) with
\(|\mathcal E_i|\le\ell\), there are at most \(L\) pairwise
stabilizer-distinct Pauli errors \(F\) with syndrome \(\sigma\)
satisfying \(\dis(F,\mathcal E)\le\rho n\).
\end{definition}

\begin{definition}[Average-radius QLD and QLR]
\label{def:average-radius-quantum}
A stabilizer code \(Q\) of block length \(n\) is
\emph{average-radius} \((\rho,L)\)-QLD if, for every syndrome
\(\sigma\) and every collection
\(F^{(1)},\ldots,F^{(L+1)}\) of pairwise stabilizer-distinct Pauli
errors with syndrome \(\sigma\),
\(
  \sum_{j=1}^{L+1}\wt(F^{(j)})>\rho(L+1)n.
\)

Likewise, \(Q\) is \emph{average-radius}
\((\rho,\ell,L)\)-QLR if, for every syndrome \(\sigma\), every
collection of local Pauli lists
\(\mathcal E_1,\ldots,\mathcal E_n\) with
\(|\mathcal E_i|\le\ell\), and every collection
\(F^{(1)},\ldots,F^{(L+1)}\) of pairwise stabilizer-distinct Pauli
errors with syndrome \(\sigma\),
\(
  \sum_{j=1}^{L+1}\dis(F^{(j)},\mathcal E)
  >\rho(L+1)n.
\)
\end{definition}

Each average-radius notion implies its corresponding standard-radius
notion.

\paragraph{Reduction to the CSS quotient.}
The definitions above are stated in terms of Pauli errors.  In the CSS
setting, the relevant information is carried entirely by their labels.

Two Pauli errors \(E_{a,b}\) and \(E_{a',b'}\) have the same syndrome
if and only if
\((a-a',b-b')\in C_1\oplus C_2\), while they are
stabilizer-equivalent if and only if
\((a-a',b-b')\in C_2^\perp\oplus C_1^\perp\).
Thus, after translating a common-syndrome family by any one reference
error, its logically distinct candidates correspond to distinct cosets
in
\[
  \mathcal N_Q/\mathcal S_Q
  =
  (C_1\oplus C_2)/(C_2^\perp\oplus C_1^\perp).
\]

More explicitly, let
\(F^{(j)}=E_{a^{(j)},b^{(j)}}\), \(j\in[L+1]\), be pairwise
stabilizer-distinct errors with a common syndrome, and take
\(F^{(1)}\) as a reference.  Define
\(c_j:=(a^{(j)}-a^{(1)},b^{(j)}-b^{(1)})\).
Then \(c_j\in\mathcal N_Q\), and the cosets
\(c_j+\mathcal S_Q\) are pairwise distinct.

The local constraints are preserved by the same translation.
Given local Pauli lists \(\mathcal E_i\), set
\(T_i:=\mathcal E_i-(a_i^{(1)},b_i^{(1)})\).
Then \(|T_i|=|\mathcal E_i|\) and
\(\dis(c_j,T)=\dis(F^{(j)},\mathcal E)\) for every \(j\).
Hence both individual and aggregate disagreement are preserved under
the translation, so the standard- and average-radius list conditions
are preserved as well.
Conversely, representatives of distinct cosets in
\(\mathcal N_Q/\mathcal S_Q\) yield pairwise stabilizer-distinct
syndrome-zero Pauli errors.

Accordingly, QLR can be viewed as a classical list-recovery problem
over distinct cosets of the logical quotient.  QLD is the special case
\(\ell=1\), in which disagreement is simply block-Hamming distance
from the prescribed center.

\subsection{Balanced FQRS and random CSS codes}
\label{sec:code-families}

\subsubsection{Balanced FQRS codes}

Let \(N<q\), let \(\alpha\) generate \(\F_q^\times\), and use the
evaluation points \(1,\alpha,\ldots,\alpha^{N-1}\).

\begin{definition}[Generalized Reed--Solomon codes;
reformulation of {\cite[Definition~3.4]{BGG24}}]
For \(k<N\) and
\(U=(u_0,\ldots,u_{N-1})\in(\F_q^\times)^N\), define
\[
  \GRS_{N,k,q}(\alpha,U)
  :=
  \left\{
    \bigl(
      u_0f(1),
      u_1f(\alpha),
      \ldots,
      u_{N-1}f(\alpha^{N-1})
    \bigr):
    f\in\F_q[x]_{<k}
  \right\}.
\]
\end{definition}

When \(U=(1,\ldots,1)\), this is the ordinary Reed--Solomon code,
which we denote by \(\RS_{N,k,q}(\alpha)\).
The dual of a GRS code is again a GRS code of complementary dimension
on the same evaluation points, with a suitable nonzero multiplier
vector~\cite[Fact~3.2]{BGG24}.

Let \(N=sn\).  The \(s\)-folding of a length-\(N\) scalar code groups
every \(s\) consecutive scalar coordinates into one symbol of
\(\F_q^s\).  We write
\(\FRS^{(s)}_{N,k,q}(\alpha)\) for the \(s\)-folding of
\(\RS_{N,k,q}(\alpha)\); folding a GRS code is defined analogously.

\begin{definition}[Balanced folded quantum RS codes;
equivalent CSS form of
{\cite[Definition~3.8 and Claim~3.7]{BGG24}}]
\label{def:fqrs}
Fix \(R\in(0,1)\), positive integers \(n,s\), and a prime power \(q\)
with \(N:=sn<q\).
Set \(R_1:=(1+R)/2\), \(k_1:=R_1N\), and
\(d_0:=(1-R_1)N\), and assume \(k_1,d_0\in\N\).
Define
\[
  C_1:=\FRS^{(s)}_{N,k_1,q}(\alpha),
  \qquad
  C_2^\perp:=\FRS^{(s)}_{N,d_0,q}(\alpha)\subseteq C_1,
  \qquad
  C_2:=(C_2^\perp)^\perp.
\]
The balanced folded quantum Reed--Solomon code of rate \(R\) is
\(\FQRS_R^{(s)}:=\CSS(C_1,C_2)\).
It is a length-\(n\) CSS code of local dimension \(q^s\), with
\(\dim_{\F_q}C_1=\dim_{\F_q}C_2=R_1sn\).
\end{definition}

In the polynomial description of~\cite[Definition~3.8]{BGG24},
the stabilizer superposition ranges over coefficients of degrees below
\(d_0\), while the logical message occupies the degrees
\([d_0,k_1)\).  Equivalently,
\(C_2^\perp=\ev(\F_q[x]_{<d_0})\subseteq
C_1=\ev(\F_q[x]_{<k_1})\), where \(\ev\) denotes folded evaluation
at the chosen points.  Thus \(C_1/C_2^\perp\) is represented by the
coefficient band \([d_0,k_1)\).

\begin{remark}[Multiplier invariance]
\label{rem:mult}
A folded GRS code differs from the corresponding folded RS code by
fixed nonzero coordinatewise multipliers.  On each folded register,
these multipliers define an invertible \(\F_q\)-linear map and hence
preserve coordinate kernels, affine dimension, Hamming disagreement,
and distinctness.  After transporting local lists by the same maps,
they also preserve list sizes and disagreement counts.  Entropy and
remainder are invariant under these bijective relabelings.
Consequently, the subspace-design, affine-containment, and
entropy/remainder statements used below transfer unchanged to the
folded GRS multiplier twist appearing in \(C_2\).
\end{remark}

\subsubsection{Balanced random CSS codes}

Following the random-CSS construction of~\cite[Appendix~A]{BGG24},
we use its full-rank formulation.

\begin{definition}[Balanced random CSS ensemble]
\label{def:random-css}
Fix \(R\in(0,1)\), set \(R_1:=(1+R)/2\), and assume
\(R_1n,(1-R_1)n\in\N\).
Sample uniformly an ordered linearly independent family
\(g_1,\ldots,g_{R_1n}\in\F_q^n\), and define
\(C_1:=\operatorname{span}\{g_1,\ldots,g_{R_1n}\}\) and
\(C_2^\perp:=
\operatorname{span}\{g_1,\ldots,g_{(1-R_1)n}\}\).
Set \(C_2:=(C_2^\perp)^\perp\).
Then \(C_2^\perp\subseteq C_1\), and
\(Q:=\CSS(C_1,C_2)\) is a balanced \([[n,Rn]]_q\) CSS code with
\(\dim C_1=\dim C_2=R_1n\).
We call this the balanced random CSS ensemble of rate \(R\).
\end{definition}

By symmetry, each of \(C_1\) and \(C_2\) is marginally a uniformly
random \(R_1n\)-dimensional subspace of \(\F_q^n\).  The two
components are not independent, since \(C_2^\perp\subseteq C_1\) by
construction.

\subsection{Subspace designs and Brascamp--Lieb inequalities}
\label{sec:designs}

For a classical code \(C\subseteq\Sigma^n\) and local lists
\(S_1,\ldots,S_n\subseteq\Sigma\), write
\(\dis(c,S):=|\{i:c_i\notin S_i\}|\).
We say that \(C\) is \((\rho,\ell,L)\)-list recoverable if every
choice of \(|S_i|\le\ell\) admits at most \(L\) codewords \(c\) with
\(\dis(c,S)\le\rho n\).  It is \emph{average-radius}
\((\rho,\ell,L)\)-list recoverable if, for every such choice of local
lists and every \(L+1\) distinct codewords
\(c^{(1)},\ldots,c^{(L+1)}\),
\(
  \sum_{j=1}^{L+1}\dis(c^{(j)},S)>\rho(L+1)n.
\)
The case \(\ell=1\) gives standard and average-radius list decoding.

\paragraph{Subspace designability.}
Let \(C\subseteq(\F_q^s)^n\) be an \(\F_q\)-linear classical code of
message dimension \(k\) and rate \(R_{\rm c}:=k/(sn)\).
Fix an injective encoder
\(\Enc_C:\F_q^k\to(\F_q^s)^n\).
For \(i\in[n]\), let \(\pi_i:\F_q^k\to\F_q^s\) be the \(i\)-th
coordinate map and set \(H_i:=\ker\pi_i\).

For an integer \(d\ge0\) and \(A\ge0\), we say that
\(H_1,\ldots,H_n\) form a \((d,A)\)-\emph{subspace design} if every
\(d\)-dimensional subspace \(U\le\F_q^k\) satisfies
\(
  \sum_i\dim(U\cap H_i)\le A.
\)

\begin{definition}[Subspace designability;
cf.~{\cite[Definition~B.2]{CZ25};
\cite[Definition~2.3 and Section~4]{BCDZ26BL}}]
\label{def:folded-subspace-designability}
For \(D\ge1\), we say that \(C\) is
\(D\)-\emph{subspace designable} if, for every
\(1\le d\le D\) and every \(d\)-dimensional subspace
\(U\le\F_q^k\),
\(
  \sum_{i=1}^n\dim(U\cap H_i)\le R_{\rm c}dn+1.
\)

For \(\mu>0\), we say that \(C\) is
\(\mu\)-\emph{slacked \(D\)-subspace designable} if, for every
\(1\le d\le D\) and every \(d\)-dimensional subspace
\(U\le\F_q^k\),
\(
  \sum_{i=1}^n\dim(U\cap H_i)\le(R_{\rm c}+\mu)dn.
\)
\end{definition}

Every \(D\)-subspace designable code is
\(\mu\)-slacked \(D\)-subspace designable whenever \(n\ge1/\mu\).

For folded Reed--Solomon codes, the following GW theorem confines the
messages in a near-capacity recovered list to a low-dimensional affine
subspace, while the GK theorem controls the corresponding coordinate
kernels.  By Remark~\ref{rem:mult}, the same statements apply to the
folded GRS multiplier twists appearing in FQRS codes.

\begin{theorem}[Guruswami--Wang affine containment;
{\cite[Theorem~7]{GW13};
see also \cite[Theorem~2.7]{BCDZ26BL}}]
\label{thm:gw}
Fix constants \(R_{\rm c}\in(0,1)\), \(0<\varepsilon<1-R_{\rm c}\), and
\(\ell\ge2\).
Let \(C\subseteq(\F_q^s)^n\) be an \(s\)-folded Reed--Solomon code
of rate \(R_{\rm c}=k/(sn)\) with appropriate evaluation points, and
suppose \(s\ge16\ell/\varepsilon^2\).
For every collection \(S_1,\ldots,S_n\subseteq\F_q^s\) with
\(|S_i|\le\ell\), the set of messages \(f\in\F_q^k\) satisfying
\(
  \dis(\Enc_C(f),S)\le(1-R_{\rm c}-\varepsilon)n
\)
is contained in an affine subspace of \(\F_q^k\) of dimension at most
\(4\ell/\varepsilon\).
\end{theorem}

\begin{theorem}[Folded-RS subspace-design bound;
{\cite[Theorem~7]{GK16};
see also \cite[Theorem~2.4]{BCDZ26BL}}]
\label{thm:gk}
Suppose \(q>sn\), \(\alpha\) generates \(\F_q^\times\), and \(C\) is
the \(s\)-folding of the Reed--Solomon code evaluated at
\(1,\alpha,\ldots,\alpha^{sn-1}\).
Let \(k\) be its message dimension and \(H_i\) its coordinate kernels.
Then, for every \(0\le d\le s\) and every \(d\)-dimensional subspace
\(U\le\F_q^k\),
\(
  \sum_{i=1}^n\dim(U\cap H_i)
  \le d(k-1)/(s-d+1).
\)
\end{theorem}

In particular, if \(C\) has classical rate
\(R_{\rm c}=k/(sn)\), \(D\le s\), and
\(s\ge(D-1)(1+R_{\rm c}/\mu)\), then \(C\) is
\(\mu\)-slacked \(D\)-subspace designable.  Indeed, for
\(1\le d\le D\),
\(
  \frac{d(k-1)}{s-d+1}
  \le
  \frac{R_{\rm c}sdn}{s-D+1}
  \le
  (R_{\rm c}+\mu)dn.
\)

\paragraph{Entropy and remainder Brascamp--Lieb inequalities.}
The two Brascamp--Lieb applications use the preceding inputs slightly
differently.  For list recovery, the GW theorem first confines a
hypothetical bad list to a low-dimensional affine message space, and
the GK theorem then supplies the subspace-design bounds needed on its
translated span.  For fixed-list decoding, no affine-containment step
is needed: after translating \(L+1\) candidates so that one is zero,
their span has dimension at most \(L\), and the GK bound applies
directly.

In either case, the subspace-design bounds verify the dimension
condition of a discrete Brascamp--Lieb inequality, yielding entropy or
remainder inequalities for every distribution on the candidate span.
For a uniform distribution, entropy is the logarithm of the number of
candidates, while the local lists control the entropies of the
coordinate projections.  Comparing these quantities gives
list-recovery bounds.  The remainder inequality plays the analogous
role for fixed-list decoding~\cite{BCDZ26BL}.

For a discrete random variable \(Y\), let \(H(Y)\) denote its Shannon
entropy and define its \emph{remainder} by
\(r(Y):=1-\max_y\Pr[Y=y]\).

% \begin{theorem}[Discrete Brascamp--Lieb inequalities ;
% equivalent form of {\cite[Theorems~1.6 and~1.7]{BCDZ26BL}}]

\begin{theorem}[Discrete entropic Brascamp--Lieb and its remainder form~\cite{CDKSY13,CDKSY24,CCE09,BCDZ26BL}]
\label{thm:functional-bl}
Let \(V\) be a finite-dimensional \(\F_q\)-vector space, and let
\(\pi_i:V\to V_i\), \(i\in[n]\), be linear maps.
Suppose there are coefficients \(s_i\ge0\) such that every subspace
\(W\le V\) satisfies
\(
  \dim W\le\sum_i s_i\dim\pi_i(W).
\)
Then every random variable \(Y\) supported on \(V\) satisfies
\[
  H(Y)\le\sum_{i=1}^n s_iH(\pi_i(Y)),
  \qquad
  r(Y)\le\sum_{i=1}^n s_i r(\pi_i(Y)).
\]
\end{theorem}

We will use the following numerical folded-RS list-recovery and
list-decoding consequences.

\begin{corollary}[Folded-RS BL inputs and list recovery;
{\cite[Corollary~4.5 and Remark~4.7]{BCDZ26BL}}]
\label{cor:frs-list-recovery}
Fix \(R_{\rm c}\in(0,1)\),
\(0<\varepsilon<1-R_{\rm c}\), and \(\ell\ge2\), and set
\(
  L:=
  \left\lfloor
    \left(
      \frac{\ell}{R_{\rm c}+\varepsilon/2}
    \right)^{3+2R_{\rm c}/\varepsilon}
  \right\rfloor.
\)
Let \(C\subseteq(\F_q^s)^n\) be an \(s\)-folded Reed--Solomon code
of rate \(R_{\rm c}\) with appropriate evaluation points.

If \(s\ge16\ell/\varepsilon^2\), then \(C\) is
\(\varepsilon/3\)-slacked
\((4\ell/\varepsilon)\)-subspace designable and is
\((1-R_{\rm c}-\varepsilon,\ell,L)\)-list recoverable.

If instead
\(
  s\ge L(1+3R_{\rm c}/\varepsilon),
\)
then \(C\) is
\(\varepsilon/3\)-slacked
\(L\)-subspace designable and is average-radius
\((1-R_{\rm c}-\varepsilon,\ell,L)\)-list recoverable.
\end{corollary}

% \begin{remark}[Entropy-to-list-recovery conversion]
% \label{rem:entropy-to-list-recovery}
% The proof of~\cite[Theorem~4.2]{BCDZ26BL} uses subspace
% designability to obtain the required entropy inequality.  The
% remaining calculation depends only on this inequality.  In
% particular, in the setting of Corollary~\ref{cor:frs-list-recovery},
% if every distribution supported on at most \(L+1\) candidates
% satisfies
% \(
%   \sum_{i=1}^n H(Y_i)
%   \ge
%   \left(1-R_{\rm c}-\frac{\varepsilon}{3}\right)nH(Y),
% \)
% then the entropy calculation in the proof of
% \cite[Theorem~4.2]{BCDZ26BL} gives average-radius
% \((1-R_{\rm c}-\varepsilon,\ell,L)\)-list recovery, and hence
% standard-radius list recovery.
% \end{remark}

\begin{theorem}[Fixed-list remainder bound;
average-radius form of {\cite[Theorem~5.2]{BCDZ26BL}}]
\label{thm:bcdz-rem}
For \(i\in[n]\), let
\(\pi_i:\F_q^k\to\F_q^s\) be linear with kernel \(H_i\), and suppose
that, for every \(d\in\{0,1,\ldots,s\}\),
\(H_1,\ldots,H_n\) form a
\(\bigl(d,d(k-1)/(s-d+1)\bigr)\) subspace design.
Let
\(C:=\{(\pi_1(x),\ldots,\pi_n(x)):x\in\F_q^k\}\).
Then, for every integer \(1\le L\le s\), \(C\) is average-radius
\((\rho,L)\)-list decodable for every
\(
  \rho<
  \frac{L}{L+1}
  \left(
    1-\frac{k-1}{n(s-L+1)}
  \right),
\)
and hence is also standard-radius \((\rho,L)\)-list decodable.
\end{theorem}

\begin{remark}[Intermediate remainder inequality]
\label{rem:bcdz-rem-intermediate}
Under the hypotheses of Theorem~\ref{thm:bcdz-rem}, fix
\(1\le L\le s\).
The proof of~\cite[Theorem~5.2]{BCDZ26BL} shows the following
intermediate statement.  If \(V\le\F_q^k\) has
\(\dim V\le L\), then every probability distribution \(X\)
supported on \(V\) satisfies
\(
  \sum_{i=1}^n r(\pi_i(X))
  \ge
  \left(
    1-\frac{k-1}{n(s-L+1)}
  \right)n\,r(X).
\)
By translation invariance of remainder, the same inequality holds
for distributions supported on any affine subspace \(a+V\) of
dimension at most \(L\).
\end{remark}

\begin{remark}[Entropy and remainder certificate conversions]
\label{rem:certificate-to-list-conversion}
The final conversion steps in the proofs of
\cite[Theorems~4.2 and~5.2]{BCDZ26BL} depend only on the corresponding
entropy or remainder inequality, and not on how that inequality is
obtained.

For the entropy case, let \(R_{\rm c},\varepsilon,\ell\), and \(L\)
be as in Corollary~\ref{cor:frs-list-recovery}.  If every distribution
supported on at most \(L+1\) candidates satisfies
\(
  \sum_{i=1}^n H(Y_i)
  \ge
  \left(1-R_{\rm c}-\frac{\varepsilon}{3}\right)nH(Y),
\)
then the entropy calculation in the proof of
\cite[Theorem~4.2]{BCDZ26BL} gives average-radius
\((1-R_{\rm c}-\varepsilon,\ell,L)\)-list recovery, and hence
standard-radius list recovery.

For the remainder case, if every distribution supported on at most
\(L+1\) candidates satisfies
\(
  \sum_{i=1}^n r(Y_i)\ge\delta n\,r(Y),
\)
then the remainder calculation in the proof of
\cite[Theorem~5.2]{BCDZ26BL} gives average-radius
\((\rho,L)\)-list decoding for every
\(
  \rho<\frac{L}{L+1}\delta,
\)
and hence standard-radius list decoding.
\end{remark}

Finally, we use the following universal lower bound for classical
linear-code list recovery.

\begin{theorem}[Linear-code list-recovery lower bound;
bad-list form of {\cite[Theorem~12]{LS25}}]
\label{thm:ls-lr-lower}
Fix \(R_{\rm c}\in(0,1)\),
\(0<\varepsilon<1-R_{\rm c}\), and an integer \(2\le\ell\le q\).
For all sufficiently large \(n\), every rate-\(R_{\rm c}\)
\(\F_q\)-linear code \(C\subseteq\F_q^n\) admits local lists
\(S_1,\ldots,S_n\subseteq\F_q\), with \(|S_i|\le\ell\), and at least
\(M\) distinct codewords \(c\in C\), where
\(M:=\ell^{\lfloor R_{\rm c}/\varepsilon\rfloor}\), satisfying
\(\dis(c,S)\le(1-R_{\rm c}-\varepsilon)n\).
Equivalently, \(C\) is not
\((1-R_{\rm c}-\varepsilon,\ell,M-1)\)-list recoverable.
\end{theorem}

\section{A unified certificate framework for CSS achievability}
\label{sec:qlr-upper}

A quantum bad list can be viewed as a classical list of paired
\(X/Z\) candidates.  A natural
sectorwise approach bounds the two projections separately and multiplies
the resulting list sizes.  For the structured and random CSS families
studied here, we instead use one-sector distributional certificates and
show that they combine without any additional loss in the certificate
coefficient.

Recall that for a discrete random variable \(Y\),
\(H(Y):=-\sum_y\Pr[Y=y]\log\Pr[Y=y]\) and
\(r(Y):=1-\max_y\Pr[Y=y]=\min_y\Pr[Y\ne y]\).
Thus \(H(Y)\) measures uncertainty, while \(r(Y)\) is the minimum
one-shot guessing error.  We use natural logarithms throughout, with
\(0\log0:=0\).

Recall also the coordinate certificates of
Definition~\ref{def:certificate}.  If \(Y\) is a random word supported
on \(C\subseteq\Sigma^n\), with \(Y_i\) its \(i\)-th coordinate random
variable, then
\(\mathsf{Cert}^{H}(\delta)\) means
\(\sum_i H(Y_i)\ge\delta n H(Y)\), while
\(\mathsf{Cert}^{r}(\delta)\) means
\(\sum_i r(Y_i)\ge\delta n r(Y)\).

To prove an output-list bound \(L\), it suffices to rule out a bad
list of \(L+1\) stabilizer-distinct candidates.  The list-decoding and
list-recovery applications therefore require certificates only for
distributions of bounded support.  For \(B\ge1\), we accordingly
define \(\mathsf{Cert}^{H}_B(\delta)\) and
\(\mathsf{Cert}^{r}_B(\delta)\) by restricting the corresponding
certificate requirements to distributions supported on at most
\(B\) words.

The pairing argument itself, however, does not require any prescribed
bound on the support size.  We therefore first prove the 
pairing lemma with no support size cap and then state the bounded-support version as an immediate
corollary.

\begin{lemma}[Pairing of coordinate certificates]
\label{lem:pairing}
Let \(\delta_X,\delta_Z\in(0,1]\), and let
\(\mathcal A\subseteq\Sigma_X^n\times\Sigma_Z^n\) be nonempty.
Define its \(X\)-projection by
\(
  \mathcal A_X:=\{x:\exists z,\ (x,z)\in\mathcal A\},
\)
and, for each \(x\in\mathcal A_X\), define the corresponding
\(Z\)-fiber by
\(
  \mathcal A_Z(x):=\{z:(x,z)\in\mathcal A\}.
\)
Set \(\delta_*:=\min\{\delta_X,\delta_Z\}\).

\begin{enumerate}
\item
If \(\mathcal A_X\) satisfies
\(\mathsf{Cert}^{H}(\delta_X)\) and every \(Z\)-fiber
\(\mathcal A_Z(x)\) satisfies
\(\mathsf{Cert}^{H}(\delta_Z)\), then every random pair
\((X,Z)\) supported on \(\mathcal A\) satisfies
\(
  \sum_{i=1}^n H(X_i,Z_i)
  \ge
  \delta_* n H(X,Z).
\)

\item
If \(\mathcal A_X\) satisfies
\(\mathsf{Cert}^{r}(\delta_X)\) and every \(Z\)-fiber
\(\mathcal A_Z(x)\) satisfies
\(\mathsf{Cert}^{r}(\delta_Z)\), then every random pair
\((X,Z)\) supported on \(\mathcal A\) satisfies
\(
  \sum_{i=1}^n r(X_i,Z_i)
  \ge
  \delta_* n r(X,Z).
\)
\end{enumerate}
\end{lemma}

\begin{proof}
We prove the entropy and remainder forms separately.

\medskip
\noindent\textbf{Entropy \(\mathsf{Cert}^{H}(\delta)\).}
By the coordinatewise chain rule,
\[
\sum_i H(X_i,Z_i)
=
\sum_i H(X_i)+\sum_i H(Z_i\mid X_i)
\ge
\sum_i H(X_i)+\sum_i H(Z_i\mid X),
\]
where the inequality holds because \(X_i\) is a deterministic function
of \(X\), and conditioning on \(X\) can only reduce entropy.
The \(X\)-certificate directly gives
\[
\sum_i H(X_i)\ge\delta_X nH(X).
\]
For every \(x\in\supp(X)\), the conditional random word
\(Z\mid X=x\) is supported on \(\mathcal A_Z(x)\), so the fiber
certificate gives
\[
\sum_i H(Z_i\mid X=x)
\ge
\delta_Z nH(Z\mid X=x).
\]
Averaging over \(x\) yields
\[
\sum_i H(Z_i\mid X)\ge\delta_Z nH(Z\mid X).
\]
Therefore
\[
\begin{aligned}
\sum_i H(X_i,Z_i)
&\ge
\delta_X nH(X)+\delta_Z nH(Z\mid X)\\
&\ge
\delta_*n\bigl(H(X)+H(Z\mid X)\bigr)\\
&=
\delta_*nH(X,Z).
\end{aligned}
\]

\paragraph{Remainder \(\mathsf{Cert}^{r}(\delta)\).}
Unlike entropy, remainder has no additive chain rule.  We instead
decompose the probability mass fiber by fiber.  
For paired words
\((x,z),(u,v)\in\Sigma_X^n\times\Sigma_Z^n\), define
\[
  d_{\rm pair}((x,z),(u,v))
  :=
  |\{i:(x_i,z_i)\ne(u_i,v_i)\}|.
\]
Thus \(d_{\rm pair}\) is the ordinary block-Hamming distance when
\(\Sigma_X\times\Sigma_Z\) is viewed as the coordinate alphabet: a
coordinate is counted once whenever either sector disagrees.

\smallskip
\noindent
\textit{Observation: remainder as minimum expected disagreement.}
Let \(Y=(Y_1,\ldots,Y_n)\) be any random word and
\(y=(y_1,\ldots,y_n)\) any deterministic center.
Since \(r(Y_i)=\min_a\Pr[Y_i\ne a]\),
\begin{equation}
\label{eq:weighted_dis_and_remainder}
\begin{aligned}
  \E[d_H(Y,y)]
  &=
  \sum_{y'\in\Sigma^n}
  \Pr[Y=y']
  \sum_{i=1}^n\mathbf 1\{y_i'\ne y_i\}\\
  &=
  \sum_{i=1}^n\Pr[Y_i\ne y_i]
  \ge
  \sum_{i=1}^n r(Y_i).
\end{aligned}
\end{equation}
Equality is attained by choosing each \(\hat y_i\) to be a most
probable value of \(Y_i\).  Thus, for a coordinatewise mode
\(\hat y=(\hat y_1,\ldots,\hat y_n)\),
\begin{equation}
\label{eq:min_expected_distance_remainder}
  \sum_{i=1}^n r(Y_i)
  =
  \min_y\E[d_H(Y,y)]
  =
  \E[d_H(Y,\hat y)].
\end{equation}

\smallskip
\noindent
\textit{Step 1: decompose into \(X\)-fibers.}
Write the support of \((X,Z)\) as
\(\{(x_j,z_j):j\in[b]\}\), with probabilities \(p_j>0\).
Fix an arbitrary paired center
\((u,v)\in\Sigma_X^n\times\Sigma_Z^n\).

Let \(x^{(1)},\ldots,x^{(h)}\) be the distinct \(X\)-words occurring
in the support, and define
\[
  J_g:=\{j:x_j=x^{(g)}\},
  \qquad
  P_g:=\sum_{j\in J_g}p_j.
\]
For fiber \(g\), let
\[
  E_g:=\{i:x_i^{(g)}\ne u_i\},
  \qquad
  E_g^c:=[n]\setminus E_g,
  \qquad
  t_g:=|E_g|.
\]
Every candidate in fiber \(g\) has the same \(X\)-word \(x^{(g)}\),
and hence the same \(X\)-disagreement
\(d_H(x^{(g)},u)=t_g\).  Therefore the total weighted
\(X\)-disagreement contributed by fiber \(g\) is
\[
  \sum_{j\in J_g}p_jt_g=P_gt_g.
\]

For the \(Z\)-sector, coordinates in \(E_g\) have already been counted
by the paired distance through the \(X\)-disagreement.  We therefore
count \(Z\)-disagreement only on \(E_g^c\).  

For
\(T\subseteq[n]\) and \(w\in\Sigma_Z^n\), define
\[
  S_g^w(T)
  :=
  \sum_{j\in J_g}p_j
  \sum_{i\in T}\mathbf 1\{z_{j,i}\ne w_i\}.
\]
It follows that the exact contribution of fiber \(g\) to the expected
paired distance is
\begin{equation}
\label{eq:D_g_exact}
  D_g
  =
  P_gt_g+S_g^v(E_g^c).
\end{equation}
Hence
\[
  \E[d_{\rm pair}((X,Z),(u,v))]
  =
  \sum_gD_g.
\]

\smallskip
\noindent
\textit{Step 2: control the residual mass inside each fiber.}
Let \(Z^{(g)}\) denote the normalized conditional distribution in
fiber \(g\). That is
\(\Pr[Z^{(g)}=z_j]=p_j/P_g\) for \(j\in J_g\).
Set 
\[
a_g:=\max_{j\in J_g}p_j.
\]
Since the paired support points are distinct, the \(z_j\)'s within a
fixed fiber are distinct. Therefore
\[
r(Z^{(g)})=1-a_g/P_g.
\]

For every \(w\in\Sigma_Z^n\),
\eqref{eq:weighted_dis_and_remainder} and the \(Z\)-fiber certificate
give
\[
  \frac1{P_g}\sum_{j\in J_g}p_jd_H(z_j,w)
  =
  \E[d_H(Z^{(g)},w)]
  \ge
  \sum_i r(Z_i^{(g)})
  \ge
  \delta_Zn\left(1-\frac{a_g}{P_g}\right).
\]
Multiplying by \(P_g\),
\begin{equation}
\label{eq:within_fiber_remainder_certificate}
  \sum_{j\in J_g}p_jd_H(z_j,w)=S_g^w(E_g)+S_g^w(E_g^c)
  \ge
  \delta_Zn(P_g-a_g).
\end{equation}

To use this inequality without double counting coordinates in \(E_g\),
we choose
\(j_g\in J_g\) with \(p_{j_g}=a_g\).
Define a patched center \(v^{(g)}\) by
\[
  v_i^{(g)}
  :=
  \begin{cases}
    z_{j_g,i}, & i\in E_g,\\
    v_i,       & i\in E_g^c.
  \end{cases}
\]
Thus, by construction, 
\(S_g^{v^{(g)}}(E_g^c)=S_g^v(E_g^c)\).

Applying \eqref{eq:within_fiber_remainder_certificate} to \(v^{(g)}\)
gives
\[
  S_g^{v^{(g)}}(E_g^c)
  \ge
  \delta_Z n(P_g-a_g)-S_g^{v^{(g)}}(E_g).
\]
We further have
\[
\begin{aligned}
S_g^{v^{(g)}}(E_g)
&=
\sum_{j\in J_g}p_j
\sum_{i\in E_g}\mathbf 1\{z_{j,i}\ne v_i^{(g)}\}
\le
(P_g-a_g)t_g.
\end{aligned}
\]
The inequality follows because, by construction, \(z_{j_g}\), of
probability mass \(a_g\), agrees with \(v^{(g)}\) on every coordinate
in \(E_g\).  Thus only the remaining mass \(P_g-a_g\) can contribute
disagreement on the \(t_g\) coordinates in \(E_g\).

Therefore
\[
  S_g^v(E_g^c)
  \ge
  \delta_Zn(P_g-a_g)-(P_g-a_g)t_g.
\]
Combining this with \eqref{eq:D_g_exact} gives
\begin{equation}
\label{eq:D_g_lower_bound}
  D_g\ge P_gt_g + \delta_Zn(P_g-a_g)-(P_g-a_g)t_g
  =
  a_gt_g+\delta_Zn(P_g-a_g).
\end{equation}

\smallskip
\noindent
\textit{Step 3: control the surviving fiber representatives.}
Summing \eqref{eq:D_g_lower_bound} over \(g\),
\[
  \E[d_{\rm pair}((X,Z),(u,v))]
  \ge
  \sum_ga_gt_g+\delta_Zn\sum_g(P_g-a_g).
\]

Set 
\[A:=\sum_ga_g,
\qquad
a_\star:=\max_ga_g=\max_jp_j.
\]
Define an auxiliary random word \(\widetilde X\), supported on the
distinct fiber representatives, by
\(\Pr[\widetilde X=x^{(g)}]=a_g/A\).
Then \(r(\widetilde X)=1-a_\star/A\).  Using
\eqref{eq:weighted_dis_and_remainder} and the \(X\)-certificate,
\[
  \frac1A\sum\nolimits_ga_gt_g
  =
  \E[d_H(\widetilde X,u)]
  \ge
  \sum\nolimits_i r(\widetilde X_i)
  \ge
  \delta_Xn\left(1-\frac{a_\star}{A}\right).
\]
Thus
\begin{equation}
\label{eq:X_weighted_fiber_bound}
  \sum_ga_gt_g
  \ge
  \delta_Xn(A-a_\star).
\end{equation}

Also,
\(\sum_g(P_g-a_g)=1-A\), while
\(r(X,Z)=1-a_\star\).
Consequently,
\begin{align}
  \E[d_{\rm pair}((X,Z),(u,v))]
  &\ge
  \delta_Xn(A-a_\star)+\delta_Zn(1-A)
  \notag\\
  &\ge
  \delta_*n\bigl((A-a_\star)+(1-A)\bigr)
  \notag\\
  &=
  \delta_*n(1-a_\star) \notag\\
  &=
  \delta_*n\,r(X,Z).
\label{eq:paired_expected_distance_lower_bound}
\end{align}

Thus the joint remainder decomposes as
\[
  r(X,Z)
  =
  \underbrace{A-a_\star}_{\text{among fiber representatives}}
  +
  \underbrace{1-A}_{\text{residual mass inside fibers}}.
\]
The first term is controlled by the \(X\)-certificate, while the
second is extracted from the fiberwise \(Z\)-certificates.

\smallskip
\noindent
\textit{Step 4: optimize over the paired center.}
The bound \eqref{eq:paired_expected_distance_lower_bound} holds for
every paired center \((u,v)\).  Choose a center
\((\hat u,\hat v)\) such that each
\((\hat u_i,\hat v_i)\) is a most probable value of
\((X_i,Z_i)\).  By Eq.~\eqref{eq:min_expected_distance_remainder},
\[
  \sum\nolimits_{i=1}^n r(X_i,Z_i)
  =
  \E[d_{\rm pair}((X,Z),(\hat u,\hat v))]
  \ge
  \delta_*n\,r(X,Z).
\]
\end{proof}

The unrestricted pairing lemma also yields a bounded-support
statement.

\begin{corollary}[Bounded-support pairing]
\label{cor:bounded-certificate-pairing}
Let \(\mathcal A,\delta_X,\delta_Z,\delta_*\) be as in
Lemma~\ref{lem:pairing}, and let \(B_X,B_Z\ge1\) be integers.
Suppose that \(\mathcal A_X\) satisfies
\(\mathsf{Cert}^{H}_{B_X}(\delta_X)\) and every \(Z\)-fiber
\(\mathcal A_Z(x)\) satisfies
\(\mathsf{Cert}^{H}_{B_Z}(\delta_Z)\).
Then every random pair \((X,Z)\) supported on \(\mathcal A\) with
\[
  |\supp(X)|\le B_X,
  \qquad
  |\supp(Z\mid X=x)|\le B_Z
  \quad\text{for every }x\in\supp(X)
\]
satisfies
\[
  \sum_{i=1}^n H(X_i,Z_i)
  \ge
  \delta_* n H(X,Z).
\]
The same statement holds with \(H\) replaced throughout by \(r\).
In particular, when \(B_X=B_Z=B\), the corresponding hypotheses
imply \(\mathsf{Cert}^{H}_B(\delta_*)\), respectively
\(\mathsf{Cert}^{r}_B(\delta_*)\), on \(\mathcal A\).
\end{corollary}

\begin{proof}
Fix a distribution satisfying the stated support bounds and set
\(\mathcal S:=\supp(X,Z)\).
Its \(X\)-projection \(\mathcal S_X\) has size at most \(B_X\),
and every \(Z\)-fiber \(\mathcal S_Z(x)\) has size at most \(B_Z\).

Every distribution on \(\mathcal S_X\) is therefore covered by
the bounded-support certificate on \(\mathcal A_X\), so
\(\mathcal S_X\) satisfies the corresponding unrestricted
certificate with coefficient \(\delta_X\).
Similarly, every \(\mathcal S_Z(x)\) satisfies the unrestricted
certificate with coefficient \(\delta_Z\).
Applying Lemma~\ref{lem:pairing} to \(\mathcal S\) proves the
claimed inequality, for either entropy or remainder.

Finally, if \(|\supp(X,Z)|\le B\), then both the projection and
fiber support bounds hold with \(B_X=B_Z=B\), giving the joint
bounded-support certificate.
\end{proof}

\paragraph{Sources of the one-sector certificates.}
The pairing lemma is agnostic to how the one-sector certificates are
obtained.  For subspace-designable codes, they follow from the
subspace-design Brascamp--Lieb inequalities; see
Section~\ref{subsec:fqrs-achievability}.  For random linear codes, the
required bounded-support certificates hold with high probability; see
Section~\ref{subsec:random_css_achievability}.  Once these one-sector
inputs are available, Lemma~\ref{lem:pairing} combines them into a
joint certificate for \((X,Z)\) with coefficient
\(\min\{\delta_X,\delta_Z\}\).

\section{Coset Singleton converse for QLD}
\label{sec:quotient-tools}\label{sec:gsb}\label{sec:qld}
For both QLR and QLD, it is sufficient to work in one sector for lower-bound purposes.  We show that a stabilizer-distinct bad list of a single sector can be lifted to a valid
joint CSS list-recovery/decoding instance of the same size.  Hence, a one-sector bad list yields the same lower bound
for the full quantum code.

Let \(Q=\CSS(C_1,C_2)\), and suppose
\(x_1,\ldots,x_M\in C_1\) form a one-sector bad list in the quotient
\(C_1/C_2^\perp\).  Thus the \(x_j\)'s are pairwise distinct modulo
\(C_2^\perp\), and there exist local lists
\(S_1,\ldots,S_n\subseteq\F_q^s\), with \(|S_i|\le \ell\), such that
each \(x_j\) disagrees with \(S=(S_1,\ldots,S_n)\) on at most
\(\rho n\) coordinates.

Lift these words to paired Pauli labels
\(
    (x_1,0),\ldots,(x_M,0).
\)
Since \(x_j\in C_1\) and \(0\in C_2\), all of these pairs lie in the
CSS normalizer and hence have syndrome zero.
Moreover, two paired candidates \((x,z)\) and \((x',z')\) are
stabilizer-equivalent exactly when
\[
    x-x'\in C_2^\perp
    \qquad\text{and}\qquad
    z-z'\in C_1^\perp .
\]
Since the \(x_j\)'s are pairwise distinct modulo \(C_2^\perp\),
the pairs \((x_1,0),\ldots,(x_M,0)\) are therefore pairwise
stabilizer-distinct.

Finally, define the joint local Pauli lists by
\(
    \widetilde E_i:=S_i\times\{0\}.
\)
Then \(|\widetilde E_i|=|S_i|\le\ell\), and for every \(j\),
\[
    (x_{j,i},0)\in \widetilde E_i
    \quad\Longleftrightarrow\quad
    x_{j,i}\in S_i.
\]
Hence the lift preserves the disagreement count exactly.  Therefore, a
bad list of size \(M\) in the marginal quotient \(C_1/C_2^\perp\)
gives a joint quantum bad list of the same size. Taking \(\ell=1\) gives the QLD case.  The lift also preserves
average-radius bad lists.

Accordingly, the QLD lower-bound analysis in this section, as well as the
QLR lower-bound constructions for FQRS and random CSS codes in
Sections~\ref{sec:fqrs-results} and~\ref{sec:random-css-results},
may be carried out entirely within one sector.

\subsection{Coset Singleton bound and its CSS consequences
(unfolded case)}

Let \(W\subseteq C\subseteq\F_q^n\) be nested linear codes.
Write the redundancy of \(C\) as
\(r_C:=n-\dim C\).
For \(g\in\F_q^n\), define the coset list
\[
  \mathcal L_{C/W}(g,\rho)
  :=
  \{\,c+W\in C/W:
      \mathcal B(g,\rho)\cap(c+W)\ne\varnothing\,\},
\]
where
\(
  \mathcal B(g,\rho)
  :=
  \{x\in\F_q^n:\dis(g,x)\le\rho n\}.
\)
Thus, a coset \(c+W\) is included whenever at least one of its
representatives lies in the Hamming ball \(\mathcal B(g,\rho)\).

\begin{lemma}[Coset generalized Singleton bound]
\label{lem:coset-gsb}
Let \(L\ge1\) with \(L+1\le |C/W|\), and put
\(
  t:=r_C+\lceil\log_q(L+1)\rceil.
\)
Then there is a center \(g\in\F_q^n\) such that
\[
  |\mathcal L_{C/W}(g,\rho_L)|\ge L+1,
  \qquad
  \rho_L
  =
  \frac1n
  \left\lceil
    \frac{L}{L+1}t
  \right\rceil .
\]
\end{lemma}

\begin{proof}
Since \(L+1\le|C/W|=q^{\dim C-\dim W}\), we have
\(\lceil\log_q(L+1)\rceil\le \dim C-\dim W\).
Choose a coordinate set \(T\subseteq[n]\) of size
\[
  |T|
  :=
  \dim C-\lceil\log_q(L+1)\rceil
  =
  n-t
  \ge
  \dim W.
\]
Every linear code \(W\le\F_q^n\) has an information set
\(I\subseteq[n]\) of size \(\dim W\) on which the projection
\(\pi_I:W\to\F_q^I\) is injective.  We may therefore choose \(T\) to
contain such an information set, so that
\(\pi_T:W\to\F_q^T\) is also injective.

Fix one representative from each coset of \(C/W\), and let
\(U\subseteq C\) denote the resulting set.  Thus
\[
  C=\bigsqcup_{u\in U}(u+W),
  \qquad
  |U|=|C/W|.
\]
Hence every \(c\in C\) has a unique representation \(c=u+w\) with
\(u\in U\) and \(w\in W\), so we may equivalently parameterize the
codewords of \(C\) by pairs in \(U\times W\).
Consider the map
\[
  U\times W\longrightarrow\F_q^T,
  \qquad
  (u,w)\longmapsto\pi_T(u+w).
\]
The domain has \(|U||W|=|C|\) pairs, while \(\F_q^T\) has only
\(q^{|T|}\) elements.  Hence some \(p\in\F_q^T\) has a preimage of size at least
\[
  \frac{|C|}{q^{|T|}}
  =
  q^{\dim C-|T|}
  =
  q^{\lceil\log_q(L+1)\rceil}
  \ge
  L+1.
\]

For any fixed \(u\in U\),
\[
  \pi_T(u+w)=\pi_T(u)+\pi_T(w).
\]
Since \(\pi_T\) is injective on \(W\), if \(w_1\ne w_2\), then
\[
  \pi_T(u+w_1)\ne \pi_T(u+w_2).
\]
Thus, for each fixed \(u\), at most one \(w\in W\) can satisfy
\(\pi_T(u+w)=p\).  Consequently, the pairs in the fiber over \(p\)
have distinct \(U\)-components.  Select \(L+1\) such pairs
\((u_1,w_1),\ldots,(u_{L+1},w_{L+1})\), and define
\(x_j:=u_j+w_j\).
Then all \(x_j\) agree with \(p\) on \(T\), while the cosets
\(x_j+W=u_j+W\) are pairwise distinct.

Let \(J:=[n]\setminus T\), so \(|J|=t\), and partition \(J\) into
\(L+1\) disjoint sets \(J_1,\ldots,J_{L+1}\) with
\(
  |J_j|\ge\lfloor t/(L+1)\rfloor.
\)
Define \(g|_T=p\) and \(g|_{J_j}=x_j|_{J_j}\).  Then \(g\) and \(x_j\)
agree on \(T\cup J_j\), so
\[
  \dis(g,x_j)
  \le
  t-\left\lfloor\frac{t}{L+1}\right\rfloor
  =
  \left\lceil\frac{L}{L+1}t\right\rceil.
\]
Thus \(\mathcal B(g,\rho_L)\) contains representatives of at least
\(L+1\) distinct cosets of \(W\).
\end{proof}

Taking \(W=\{0\}\) recovers the usual generalized Singleton bound for
ordinary linear codes.  More generally, once the radius exceeds
\(r_C/n\), the forced coset list grows exponentially in the excess
\(w-r_C\), as quantified by the following proposition.

\begin{proposition}[Exponential coset-list growth past \(r_C\)]
\label{prop:blowup}
Fix \(\rho\in[0,1]\) and put \(w:=\lfloor\rho n\rfloor\).  If
\(w\ge r_C+1\), then
\[
  \max_{g\in\F_q^n}
  |\mathcal L_{C/W}(g,\rho)|
  \ge
  \min\bigl\{
    q^{\,w-r_C},
    |C/W|
  \bigr\}.
\]
\end{proposition}

\begin{proof}
If \(W=C\), then \(C/W\) contains a single coset, namely \(C\) itself.
Taking \(g=0\), we have \(0\in C\cap\mathcal B(0,\rho)\), so
\(|\mathcal L_{C/W}(0,\rho)|=1=|C/W|\), and the claim follows.

Otherwise \(W\ne C\).  Since \(w\ge r_C+1\), both
\(w-r_C\) and \(\dim C-\dim W\) are at least \(1\).  Set
\[
  m:=\min\{w-r_C,\,\dim C-\dim W\}\ge1.
\]
Let
\(
  L+1:=q^m.
\)
Then \(L+1\le q^{\dim C-\dim W}=|C/W|\), and
Lemma~\ref{lem:coset-gsb} applies with
\(\lceil\log_q(L+1)\rceil=m\) and \(t=r_C+m\le w\).  It gives a
center \(g\) with at least \(q^m\) distinct cosets inside radius
\[
  \rho_0
  = \frac1n
  \left\lceil
    \frac{L}{L+1}t
  \right\rceil
  =
  \frac1n
  \left\lceil
    (1-q^{-m})(r_C+m)
  \right\rceil.
\]
Since
\[
  \left\lceil(1-q^{-m})(r_C+m)\right\rceil
  \le
  r_C+m
  \le
  w
  =
  \lfloor\rho n\rfloor
  \le
  \rho n,
\]
we have \(\rho_0\le\rho\).  By monotonicity of the coset list in the
radius,
\[
  |\mathcal L_{C/W}(g,\rho)|
  \ge
  |\mathcal L_{C/W}(g,\rho_0)|
  \ge
  q^m
  =
  \min\{q^{\,w-r_C},|C/W|\}.
\]
\end{proof}

The proposition shows that the forced coset list grows as
\(q^{\,w-r_C}\), up to the size of the quotient, where the excess
\(w-r_C\) may scale differently with \(n\).  In particular, if \(C\)
has rate \(R_\mathrm{c}:=\dim C/n\) and
\(\rho=1-R_\mathrm{c}+\varepsilon\) for a constant \(\varepsilon>0\), then
\(w-r_C=\varepsilon n+O(1)\), and hence the proposition gives a
worst-case coset list of size at least
\(
  \min\{q^{\varepsilon n+O(1)},\,|C/W|\}.
\)
% Thus a constant excess in the relative radius forces exponential growth
% in \(n\).

We now apply Proposition~\ref{prop:blowup} to the two marginal quotients
of a CSS code, obtaining a direct constraint on their redundancies.

\begin{corollary}[Marginal blow-up and balance barrier]
\label{cor:css-balance-barrier}
Let \(Q=\CSS(C_1,C_2)\) range over a fixed-rate family of scalar CSS
codes of rate \(R=k/n\in(0,1)\), and write
\(
  r_X:=n-\dim C_1,~
  r_Z:=n-\dim C_2,~
  r_{\min}:=\min\{r_X,r_Z\}.
\)
Suppose the family is
\((\rho,L_n)\)-QLD with \(\log_qL_n=o(n)\), where
\(
  \rho=(1-R)/2-\gamma.
\)
Then
\[
  r_{\min}\ge \rho n-o(n),
\]
and hence
\[
  |r_X-r_Z|
  \le
  2\gamma n+o(n).
\]
\end{corollary}

\begin{proof}
Assume without loss of generality that \(r_{\min}=r_X\), and put
\(w:=\lfloor\rho n\rfloor\).
Since \(k=\dim C_1+\dim C_2-n=Rn\), we have
\[
  |C_1/C_2^\perp|
  =
  q^{\dim C_1-\dim C_2^\perp}
  =
  q^{\dim C_1+\dim C_2-n}
  =
  q^k.
\]
Because \(R>0\) is fixed and \(\log_qL_n=o(n)\), we have
\(L_n<q^k\) for all sufficiently large \(n\).

We claim that, for such \(n\),
\(
  r_X\ge w-\log_qL_n.
\)
If \(w\le r_X\), this is immediate since \(L_n\ge1\).
Otherwise \(w\ge r_X+1\), so
Proposition~\ref{prop:blowup}, applied with
\(C=C_1\) and \(W=C_2^\perp\), gives a one-sector coset list of
size at least \(\min\{q^{w-r_X},q^k\}\).
By the one-sector lift above, this is also a lower bound on the
joint quantum list size.  Hence
\[
  \min\{q^{w-r_X},q^k\}\le L_n.
\]
Since \(L_n<q^k\), this forces \(q^{w-r_X}\le L_n\), or
equivalently \(r_X\ge w-\log_qL_n\), proving the claim.

Consequently,
\[
  r_{\min}
  =
  r_X
  \ge
  \lfloor\rho n\rfloor-\log_qL_n
  \ge
  \rho n-1-\log_qL_n
  =
  \rho n-o(n).
\]
Finally, using \(r_X+r_Z=(1-R)n\) and
\(\rho=(1-R)/2-\gamma\), we obtain
\[
\begin{aligned}
  |r_X-r_Z|
  &=
  r_Z-r_X\\
  &=
  (1-R)n-2r_X\\
  &\le
  (1-R)n-2\lfloor\rho n\rfloor+2\log_qL_n\\
  &\le
  2\gamma n+2\log_qL_n+2\\
  &=
  2\gamma n+o(n).
\end{aligned}
\]
\end{proof}

% \begin{proof}
% Assume without loss of generality that \(r_{\min}=r_X\).
% Note that
% \(
%   k
%   =
%   n-\dim(C_2^\perp\oplus C_1^\perp)
%   =
%   \dim C_1+\dim C_2-n.
% \)
% Apply Proposition~\ref{prop:blowup} with
% \(C=C_1\) and \(W=C_2^\perp\).  Since
% \[
%   |C_1/C_2^\perp|
%   =
%   q^{\dim C_1-\dim C_2^\perp}
%   =
%   q^{\dim C_1+\dim C_2-n}
%   =
%   q^k,
% \]
% whenever \(w:=\lfloor\rho n\rfloor\ge r_X+1\), the proposition gives
% a one-sector coset list of size at least
% \(
%   \min\{q^{\,w-r_X},q^k\}.
% \)
% By the one-sector observation above, this is also a lower bound on the
% joint quantum list size.

% We now show that subexponential QLD list size forces
% \(
%   r_X\ge\rho n-o(n).
% \)
% If \(w\le r_X\), this follows immediately from
% \(w=\lfloor\rho n\rfloor=\rho n+O(1)\).
% Otherwise \(w\ge r_X+1\), so
% \(
%   \min\{q^{\,w-r_X},q^k\}\le L_n.
% \)
% Since \(k=Rn\) with fixed \(R>0\), whereas
% \(\log_qL_n=o(n)\), we have \(L_n<q^k\) for all sufficiently large
% \(n\).  Hence
% \(
%   q^{\,w-r_X}\le L_n,
% \)
% and therefore
% \(
%   w-r_X\le\log_qL_n=o(n).
% \)
% Again using \(w=\rho n+O(1)\), we obtain
% \(
%   r_X\ge\rho n-o(n).
% \)

% Finally, since
% \(r_X+r_Z=(1-R)n\) and
% \(\rho=(1-R)/2-\gamma\),
% \[
%   r_Z-r_X
%   =
%   (1-R)n-2r_X
%   \le
%   2\gamma n+o(n).
% \]
% \end{proof}

Thus, near the quantum Singleton radius, any scalar CSS family with
subexponential QLD list size must be approximately balanced.
We next specialize to balanced CSS codes.

\begin{corollary}[Balanced CSS consequence]
\label{cor:balanced-css-qld-lower}
Fix \(R\in(0,1)\), and let \(Q=\CSS(C_1,C_2)\) range over a
family of balanced scalar CSS codes of rate \(R\).
For every fixed integer \(L\ge1\) and all sufficiently large
admissible \(n\), there is a QLD bad list of size \(L+1\) at radius
\[
  \rho_L
  =
  \frac{L}{L+1}\frac{1-R}{2}+O(1/n).
\]
Consequently, for fixed \(\gamma\in(0,(1-R)/2)\), any
blocklength-independent list budget \(L\) for which the family is
\(
  ((1-R)/2-\gamma,L)
\)-QLD for all sufficiently large admissible \(n\) must satisfy
\[
  L\ge\frac{1-R}{2\gamma}-1.
\]
In particular, \(L=\Omega((1-R)/\gamma)\) as
\(\gamma\downarrow0\), with \(R\) fixed.
The same conclusions hold for both standard- and average-radius QLD.
\end{corollary}

\begin{proof}
For a balanced CSS code,
\(
  r_X=r_Z=(1-R)n/2.
\)
Moreover,
\(|C_1/C_2^\perp|=q^k=q^{Rn}\), so for every fixed \(L\),
\(L+1\le |C_1/C_2^\perp|\) for all sufficiently large \(n\).
Apply Lemma~\ref{lem:coset-gsb} to the marginal quotient
\(C_1/C_2^\perp\).  Since
\(\lceil\log_q(L+1)\rceil=O(1)\) for fixed \(L\), the lemma gives a bad
list of size \(L+1\) at radius
\[
  \rho_L
  =
  \frac1n
  \left\lceil
    \frac{L}{L+1}
    \left(
      \frac{1-R}{2}n+\lceil\log_q(L+1)\rceil
    \right)
  \right\rceil
  =
  \frac{L}{L+1}\frac{1-R}{2}+O(1/n).
\]
By the one-sector observation above, this is also a bad list for the
full CSS code.

Now suppose the code is QLD at radius
\(
  \rho=(1-R)/2-\gamma
\)
with blocklength-independent list size \(L\).  Since a bad list of
size \(L+1\) exists at radius \(\rho_L\), we must have
\(\rho<\rho_L\) for all sufficiently large \(n\).  Hence
\[
  \frac{1-R}{2}-\gamma
  \le
  \frac{L}{L+1}\frac{1-R}{2}+o(1).
\]
Letting \(n\to\infty\) gives
\(
  \gamma\ge (1-R)/(2(L+1)),
\)
and therefore
\(
  L\ge(1-R)/(2\gamma)-1.
\)
In particular,
\(
  L=\Omega((1-R)/\gamma).
\)

The same construction also witnesses the average-radius lower bound,
since every one of the \(L+1\) representatives lies within radius
\(\rho_L\), and hence their average disagreement is at most
\(\rho_L n\).
\end{proof}

\subsection{Folded quotients}
\label{subsec:folded-fqrs-quotients}

The scalar proof of Lemma~\ref{lem:coset-gsb} chooses a projection set
\(T\) containing an information set of the stabilizer subcode \(W\).
Injectivity of \(\pi_T\) on \(W\) then ensures that, for any fixed coset
representative \(u\), at most one word in \(u+W\) can lie in a given
projection fiber.  Thus \(L+1\) codewords in the same fiber necessarily
represent \(L+1\) distinct cosets of \(W\).

After folding, the code is viewed as a length-\(n\) code over the register
alphabet \(\F_q^s\), and Hamming distance is measured in whole registers.
The underlying code, however, remains only \(\F_q\)-linear.  An
\(\F_q\)-information set for \(W\) may therefore be scattered across many
folded registers and need not correspond to the minimum possible number of
whole registers.

For an \(\F_q\)-linear subspace \(W\le(\F_q^s)^n\), define its
\emph{register information number} by
\[
  \operatorname{ris}_s(W)
  :=
  \min\bigl\{
    |I|: I\subseteq[n]
    \text{ and }
    \pi_I:W\to(\F_q^s)^I
    \text{ is injective}
  \bigr\}.
\]
Since \((\F_q^s)^I\) has \(\F_q\)-dimension \(s|I|\), necessarily
\(
  \operatorname{ris}_s(W)
  \ge
  \left\lceil\frac{\dim W}{s}\right\rceil.
\)

To apply the scalar projection-and-patching argument over whole registers,
we need \(T\subseteq[n]\) to be large enough that \(\pi_T\) is injective on
\(W\), but small enough that some projection fiber still contains at least
\(L+1\) codewords.  Thus it suffices that
\[
  \operatorname{ris}_s(W)
  \le |T|
  \le
  \frac{\dim C}{s}-\log_{q^s}(L+1).
\]
Equivalently, the projection-and-patching argument requires enough whole
registers to determine \(W\), but few enough registers that some fiber still
contains \(L+1\) codewords.

In the scalar case \(s=1\),
\(\operatorname{ris}_1(W)=\dim W\), so these two requirements are compatible
whenever
\(
  L+1\le |C/W|.
\)
After folding, this compatibility is not automatic since the independent scalar
directions of \(W\) may be spread across more than
\(\lceil\dim_{\F_q}W/s\rceil\) registers. We therefore propose a conditional generalized Singleton bound for folded codes as follows.

\begin{lemma}[Conditional Folded coset generalized Singleton bound]
\label{lem:folded-coset-gsb}
Let \(W\subseteq C\subseteq(\F_q^s)^n\) be
\(\F_q\)-linear, and write \(k_C:=\dim C\).
Define \(\mathcal L_{C/W}(g,\rho)\) using Hamming distance
on the \(n\) folded registers.
For an integer \(L\ge1\), set
\(
  m_L:=
  \left\lfloor
    \frac{k_C}{s}-\log_{q^s}(L+1)
  \right\rfloor.
\)
If \(\operatorname{ris}_s(W)\le m_L\), then there exists
\(g\in(\F_q^s)^n\) such that
\[
  |\mathcal L_{C/W}(g,\rho_L)|\ge L+1,
  \qquad
  \rho_L=
  \frac1n
  \left\lceil
    \frac{L}{L+1}(n-m_L)
  \right\rceil.
\]
\end{lemma}

\begin{proof}
We apply the argument of Lemma~\ref{lem:coset-gsb} to whole
registers.  The hypothesis allows us to choose \(T\subseteq[n]\)
with \(|T|=m_L\) such that \(\pi_T\) is injective on \(W\).
Since
\(
  |C|/q^{s|T|}
  =
  q^{k_C-sm_L}
  \ge L+1,
\)
the same pigeonhole argument gives \(L+1\) quotient-distinct
codewords agreeing on \(T\).  The same patching step on the
remaining \(n-m_L\) registers then gives the claimed radius.
\end{proof}

For FQRS codes, the MDS structure of the underlying stabilizer
subcode gives
\(
  \operatorname{ris}_s(C_2^\perp)
  =
  \lceil\dim_{\F_q}(C_2^\perp)/s\rceil.
\)
% Under the block-alignment assumption used in the FQRS application,
% the condition \(L+1\le|C_1/C_2^\perp|\) therefore guarantees the
% hypothesis of Lemma~\ref{lem:folded-coset-gsb}.
Random CSS codes in this paper are scalar, so their converse follows
directly from Lemma~\ref{lem:coset-gsb}.

\section{Folded quantum Reed--Solomon codes}
\label{sec:fqrs-results}

We consider the balanced FQRS codes of
Definition~\ref{def:fqrs}.  Throughout this section, fix the quantum
rate \(R\in(0,1)\), and write
\(
  R_1:=\frac{1+R}{2}
\)
for the common component rate.  
% For
% \(0<\gamma<(1-R)/2\), set
% \(
%   \rho:=\frac{1-R}{2}-\gamma
%   =1-R_1-\gamma.
% \)

Write
\(
  Q=\mathrm{FQRS}^{(s)}_R=\CSS(C_1,C_2).
\)
Here \(C_1\) and \(C_2^\perp\) are standard folded RS codes on the
same evaluation points, while \(C_2\) is a folded GRS code arising
from duality.  The achievability results use entropy and remainder
certificates for QLR and QLD, respectively.  The QLR lower bound is
obtained by an explicit construction in the quotient
\(C_1/C_2^\perp\), while the QLD converse uses the coset Singleton
tools of Section~\ref{sec:quotient-tools}.

\subsection{Achievability: Classical BL inputs}
\label{subsec:fqrs-achievability}

We first show how subspace designability provides the one-sector
entropy and remainder certificates needed by the pairing
lemma.  We then apply this input to standard folded
RS codes. By Remark~\ref{rem:mult}, the same conclusions hold for the
corresponding folded GRS sector.

\begin{lemma}[Subspace-design BL certificates]
\label{lem:design-to-certificate}
Let \(C\subseteq(\F_q^s)^n\) be a
\(\mu\)-slacked \(D\)-subspace-designable code of rate \(R_{\rm c}\), with
\(R_{\rm c}+\mu<1\), and let
\(\Enc_C:\F_q^k\to C\) be its encoder.
Then, for every affine subspace \(a+V\subseteq\F_q^k\) with
\(\dim V\le D\), the encoded set
\(
  \Enc_C(a+V)\subseteq C
\)
satisfies
\(
  \mathsf{Cert}^{H}(1-R_{\rm c}-\mu)
\)
and
\(
  \mathsf{Cert}^{r}(1-R_{\rm c}-\mu).
\)
\end{lemma}

\begin{proof}
Fix an affine subspace \(a+V\subseteq\F_q^k\) with
\(\dim V\le D\).  Since \(C\) is linear,
\(
  \Enc_C(a+V)=\Enc_C(a)+\Enc_C(V).
\)
Coordinatewise translation by \(\Enc_C(a)\) preserves entropy and
remainder, so it suffices to prove the certificates on
\(\Enc_C(V)\).

For every subspace \(W\le V\), subspace designability and
rank--nullity give
\[
  \sum_{i=1}^n\dim\pi_i(W)
  =
  n\dim W-\sum_{i=1}^n\dim(W\cap H_i)
  \ge
  (1-R_{\rm c}-\mu)n\dim W.
\]
Thus Theorem~\ref{thm:functional-bl}, applied to \(V\) and the
restricted coordinate maps \(\pi_i|_V\), gives the claimed
certificates on \(\Enc_C(V)\), and hence on \(\Enc_C(a+V)\).
\end{proof}

\begin{theorem}[FQRS QLR]
\label{thm:fqrs-qlr}
Fix \(\ell\ge2\) and \(\gamma\in(0,(1-R)/2)\), and set
\(
  L:=
  \left\lfloor
    \left(\frac{\ell}{R_1+\gamma/2}\right)^{3+2R_1/\gamma}
  \right\rfloor.
\)
If \(s\ge16\ell/\gamma^2\), then \(Q\) is standard-radius
\(((1-R)/2-\gamma,\ell,L)\)-QLR.  If instead
\(s\ge L(1+3R_1/\gamma)\), then \(Q\) is average-radius
\(((1-R)/2-\gamma,\ell,L)\)-QLR.

% In particular, under the corresponding folding condition, the required
% blocklength-independent list size satisfies
% \[
%   L=\ell^{O(R_1/\gamma)}.
% \]
\end{theorem}

% \begin{proof}
% The one-sector entropy-BL calculations underlying
% Corollary~\ref{cor:frs-list-recovery} give the stated standard- and
% average-radius parameters at classical sector rate \(R_1\) and slack
% \(\varepsilon=\gamma\).
% Lemma~\ref{lem:pairing} combines the two sector entropy certificates
% with no loss in the coefficient, so the same calculation applies
% unchanged to the joint paired list. 
% \end{proof}

\begin{proof}
Set \(\delta:=1-R_1-\gamma/3\) and focus on one sector.

For standard-radius list recovery, Theorem~\ref{thm:gw} guarantees that
each relevant candidate set is contained in an affine message space of
dimension at most \(4\ell/\gamma\).
Corollary~\ref{cor:frs-list-recovery} then gives
\(\gamma/3\)-slacked subspace designability through this dimension.
For average-radius list recovery, given \(L+1\) paired candidates,
each one-sector projection or fiber has affine message span of dimension
at most \(L\).  The same corollary gives
\(\gamma/3\)-slacked \(L\)-subspace designability under the stated
folding condition.

Thus Lemma~\ref{lem:design-to-certificate} gives
\(\mathsf{Cert}^{H}(\delta)\) on the required one-sector candidate
sets in either case.  By Remark~\ref{rem:mult}, the same conclusion
holds for the multiplier-twisted sector.

The pairing Lemma~\ref{lem:pairing} then gives the same entropy
certificate for the paired \(X/Z\) candidates.
Applying Remark~\ref{rem:certificate-to-list-conversion} with
\(R_{\rm c}=R_1\) and \(\varepsilon=\gamma\) gives the stated
list size and radius.
\end{proof}

\begin{theorem}[FQRS QLD achievability]
\label{thm:fqrs-qld}
Let \(Q\) be a balanced FQRS code of quantum rate \(R\), and let
\(R_1=(1+R)/2\).  For every integer \(1\le L\le s\), \(Q\) is
average-radius \((\rho,L)\)-QLD for every
\[
  \rho<
  \frac{L}{L+1}
  \left(
    1-\frac{k_1-1}{n(s-L+1)}
  \right),
  \qquad k_1=R_1sn.
\]
% In particular, for fixed \(L\), the achievable radius approaches
% \(\frac{L}{L+1}(1-R_1)
% =\frac{L}{L+1}\frac{1-R}{2}\)
% as the folding parameter \(s\) grows.

\end{theorem}

\begin{proof}
Set
\(
  \delta_L
  :=
  1-(k_1-1)/(n(s-L+1)).
\)
For any \(L+1\) paired candidates, each one-sector projection or fiber
is contained in an affine message space of dimension at most \(L\).
Remark~\ref{rem:bcdz-rem-intermediate} therefore gives the remainder
inequality with coefficient \(\delta_L\) on the required one-sector
candidate sets.  The remaining steps are analogous to the QLR argument above.
\end{proof}

\subsection{Quotient-surviving lower bounds}

\subsubsection{QLR: Diagonal-product grid construction}
\label{subsec:fqrs-soft-grid}

Classically, in the relevant large-alphabet linear-code regime, a code
of rate \(R_1\) at radius \(1-R_1-\eps\) requires output list size
\(L\ge \ell^{\Omega(R_1/\eps)}\) for list
recovery~\cite{LS25,BCDZ26BL}.  However, such a classical lower bound
does not immediately imply a quantum one: distinct classical codewords
may become equivalent modulo the stabilizer and collapse to a smaller
quantum list.  Our goal is therefore to construct a large classical bad
list whose candidates remain distinct after stabilizer quotienting.  As
discussed above, for lower-bound purposes it is sufficient to work in one
sector. 

Our construction builds on the folded-RS lower-bound construction of
Chen and Zhang~\cite{CZ25}.  Their construction arranges a large
product-type family of codewords so that each coordinate takes only a
small number of possible symbols, thereby producing a large classical
bad list for list recovery.  Since the construction is classical,
however, it does not address stabilizer equivalence and does not by
itself guarantee that distinct candidates remain distinct modulo
\(C_2^\perp\).

We adapt their construction to the quantum quotient by modifying the
generator polynomials to have distinct degrees.  This ensures that the
highest-degree term in any nonzero difference of two candidates cannot
cancel, and hence allows us to enforce stabilizer distinctness.  For
completeness, we give the full construction below, adapted to our
notation.

% For a stabilizer code \(Q\), let
% \[
%   L_Q^\star(\rho,\ell)
%   :=
%   \min\{\,L:Q\text{ is }(\rho,\ell,L)\QLR\,\}.
% \]

For FQRS, we work in the quotient \(C_1/C_2^\perp\).  As in
Definition~\ref{def:fqrs}, \(C_1\) and \(C_2^\perp\) are obtained by
folding ordinary Reed--Solomon codes on the same scalar evaluation
points
\(
  1,\alpha,\ldots,\alpha^{sn-1},
\)
where \(\alpha\) is a generator of \(\F_q^\times\). For simplicity, further assume that \(R_1n\in\mathbb N\), so that
\[
  k_1:=\dim_{\F_q} C_1 = R_1sn,
  \qquad
  d_0:=\dim_{\F_q} C_2^\perp = (1-R_1)sn
\]
are both multiples of \(s\).
Then
\[
  C_2^\perp
  =
  \ev(\F_q[x]_{<d_0})
  \subseteq
  C_1
  =
  \ev(\F_q[x]_{<k_1}).
\]

Since \(k_1<sn\) and the \(sn\) evaluation points are distinct, the
folded evaluation map is injective on
\(\F_q[x]_{<k_1}\).  Hence every \(y\in C_1\) has a unique encoding
polynomial \(f_y\in\F_q[x]_{<k_1}\).  Under this identification,
\[
  y\in C_2^\perp
  \iff
  \deg f_y<d_0.
\]
Consequently, for \(y_j=\ev(f_j)\in C_1\),
\[
  y_1\equiv y_2\pmod{C_2^\perp}
  \iff
  f_1-f_2\in\F_q[x]_{<d_0}
  \iff
  \deg(f_1-f_2)<d_0.
\]
Thus, to obtain stabilizer-distinct candidates, it suffices to construct
a classical bad list in \(C_1\) such that every nonzero pairwise
polynomial difference has degree at least \(d_0\).

\noindent\textbf{\textit{Step 1: diagonal generators.}}
Fix \(0<\gamma<(1-R)/2\) and an integer \(2\le\ell\le q\), and set
\(b:=\lceil R_1/\gamma\rceil\). For \(i\in[n]\) and \(m\in[s]\),
write
\(
  \alpha_{i,m}:=\alpha^{\,s(i-1)+m-1}
\)
for the \(m\)-th scalar evaluation point in folded register \(i\).
Thus the \(i\)-th folded symbol of a polynomial \(f\) is
\(
  \pi_i(f):=
  \bigl(f(\alpha_{i,1}),\ldots,f(\alpha_{i,s})\bigr).
\)

Fix an integer \(r\ge1\) with \(br\le n\), and choose pairwise
disjoint sets \(G_1,\ldots,G_b\subseteq[n]\), each of size \(r\);
further conditions on \(r\) are specified below.
For each \(j\in[b]\), define
\[
  u_j(x)
  :=
  x^{j-1}
  \prod_{\substack{h\in[b]\\h\ne j}}
  \prod_{i\in G_h}
  \prod_{m=1}^{s}
  (x-\alpha_{i,m}).
\]
The product vanishes on every group \(G_h\) with \(h\ne j\).
On \(G_j\), every factor is nonzero because the scalar evaluation
points are pairwise distinct and nonzero.  Hence
\begin{equation}
\label{eq:uj_pattern}
\begin{aligned}
  \pi_i(u_j)&=0
  &&\text{for every }i\in G_h,\ h\ne j,\\
  \pi_i(u_j)&\ne0
  &&\text{for every }i\in G_j.
\end{aligned}
\end{equation}
The factor \(x^{j-1}\) separates the generator degrees without
changing this vanishing pattern.  Each \(u_j\) is monic, with
\begin{equation}
\label{eq:deg_uj_range}
  s(b-1)r
  \le
  \deg u_j
  =
  s(b-1)r+j-1
  <
  s(b-1)r+b,
  \qquad j\in[b].
\end{equation}
Thus the \(u_j\)'s have pairwise distinct, strictly increasing degrees
and are therefore linearly independent: in any nonzero linear
combination, the highest-degree term cannot cancel.

Let \(\rho:=(1-R)/2-\gamma\).  For the moment, suppose that \(r\) can
be chosen so that
\begin{equation}
\label{eq:grid-window}
\begin{aligned}
  \textnormal{(i) } & br\le n,\qquad
  \textnormal{(ii) } s(b-1)r+b\le k_1,\qquad
  \textnormal{(iii) } s(b-1)r\ge d_0,\qquad
  \textnormal{(iv) } br\ge(1-\rho)n.
\end{aligned}
\end{equation}
We verify below that such an \(r\) exists.

Condition~(i) ensures that the \(b\) disjoint groups fit inside the
\(n\) registers.  Condition~(ii) ensures that the largest generator
degree, \(s(b-1)r+b-1\), is strictly below \(k_1\), so every
\(\ev(u_j)\) lies in \(C_1\).  Condition~(iii) places every generator
degree at or above the stabilizer cutoff \(d_0\); together with the
distinct generator degrees, this will ensure that every nonzero
difference between the candidates constructed below lies outside
\(C_2^\perp\).  Finally, Condition~(iv) ensures that the designated
groups cover at least \((1-\rho)n\) registers, so the resulting
candidates satisfy the radius-\(\rho\) list-recovery condition.

\noindent\textbf{\textit{Step 2: an \(\ell^b\)-point product grid.}}
We now use the generators \(u_1,\ldots,u_b\) to construct
\(\ell^b\) distinct candidates and a single list-recovery instance
\(S=(S_1,\ldots,S_n)\) that contains all of them within radius
\(\rho\).

Fix any \(T\subseteq\F_q\) with \(|T|=\ell\).  For each
\(t=(t_1,\ldots,t_b)\in T^b\), define the polynomial \(f_t\) and its
corresponding codeword candidate \(y_t\) by
\[
  f_t:=\sum\nolimits_{j=1}^b t_j u_j,
  \qquad
  y_t:=\operatorname{ev}(f_t)\in C_1.
\]
Since \(u_1,\ldots,u_b\) are linearly independent, distinct
\(t,t'\in T^b\) give distinct polynomials \(f_t\ne f_{t'}\).  Hence
there are \(|T^b|=\ell^b\) distinct polynomials in the family.

By Condition~\eqref{eq:grid-window}(ii), every \(f_t\) has degree
strictly less than \(k_1<sn\).  Since the \(sn\) scalar evaluation
points are distinct, evaluation is injective on
\(\F_q[x]_{<k_1}\).  Consequently, the \(\ell^b\) distinct
polynomials \(f_t\) give exactly \(\ell^b\) distinct codewords
\(y_t\in C_1\).

Define a single list-recovery instance \(S=(S_1,\ldots,S_n)\) by
\[
  S_i:=
  \begin{cases}
    \{\,c\,\pi_i(u_j):c\in T\,\},
      & i\in G_j,\\[2pt]
    \{0\},
      & i\notin G_1\cup\cdots\cup G_b.
  \end{cases}
\]
Clearly \(|S_i|\le\ell\) for every \(i\).

For any \(i\in G_1\cup\cdots\cup G_b\), let \(j\in[b]\) be the unique
index such that \(i\in G_j\).  By~\eqref{eq:uj_pattern},
\[
  \pi_i(f_t)
  =
  \sum\nolimits_{h=1}^b t_h\pi_i(u_h)
  =
  t_j\pi_i(u_j)
  \in S_i.
\]
Thus every candidate \(y_t\) agrees with the local lists on every
register in \(G_1\cup\cdots\cup G_b\).  Since these groups are
pairwise disjoint and each has size \(r\), this gives agreement on
\(br\) registers.  By~\eqref{eq:grid-window}(iv), for every
\(t\in T^b\),
\[
  \operatorname{dis}(y_t,S)
  \le
  n-br
  \le
  \rho n.
\]
Therefore all \(\ell^b\) codewords lie simultaneously in the same
radius-\(\rho\) list-recovery instance.

\noindent\textbf{\textit{Step 3: no collapse modulo the stabilizer.}}
It remains to verify that these \(\ell^b\) classical candidates remain
distinct after quotienting by \(C_2^\perp\).  Take distinct
\(t,t'\in T^b\), and let
\[
  j_\star:=\max\{j:t_j\ne t'_j\},
\]
the largest index at which the two coefficient vectors differ.  Then
\[
  f_t-f_{t'}
  =
  \sum\nolimits_{j=1}^b (t_j-t'_j)u_j.
\]
Since \(t_{j_\star}-t'_{j_\star}\ne0\) and the \(u_j\)'s have strictly
increasing degrees, \(u_{j_\star}\) is the highest-degree generator
appearing with nonzero coefficient in this sum.  Its leading term
therefore cannot cancel, and hence
\[
  \deg(f_t-f_{t'})
  =
  s(b-1)r+j_\star-1.
\]
By~\eqref{eq:grid-window}(ii)--(iii),
\[
  d_0
  \le
  \deg(f_t-f_{t'})
  <
  k_1.
\]
Therefore
\[
  y_t-y_{t'}
  =
  \operatorname{ev}(f_t-f_{t'})
  \notin C_2^\perp.
\]
Thus the \(\ell^b\) candidates remain pairwise distinct modulo the
stabilizer.

Therefore, we obtain \(\ell^b\) pairwise
stabilizer-distinct quantum candidates in a single radius-\(\rho\)
list-recovery instance.  Since each candidate individually has
disagreement at most \(\rho n\), the same construction gives both
standard- and average-radius QLR lower bounds, provided that an integer
\(r\) satisfying~\eqref{eq:grid-window} exists.

\begin{theorem}[FQRS QLR lower bound]
\label{thm:fqrs-grid}
Fix \(R\in(0,1)\), \(\gamma\in(0,(1-R)/2)\), and \(2\le\ell\le q\), and assume \(R_1n\in\mathbb N\).
Put \(R_1:=(1+R)/2\), \(\rho:=(1-R)/2-\gamma\), and
\(b:=\lceil R_1/\gamma\rceil\).  Define
\(
  \lambda
  :=
  \max\left\{
    \frac{1-\rho}{b},
    \frac{1-R_1}{b-1}
  \right\}.
\)
Then, for \(r:=\lceil\lambda n\rceil\), the four inequalities in
\eqref{eq:grid-window} hold whenever
\[
  n\ge
  \max\left\{
    \frac{1}{1/b-\lambda},\,
    \frac{s(b-1)+b}
         {R_1s-s(b-1)\lambda}
  \right\}.
\]
Consequently, for every admissible \((n,s,q)\) satisfying this bound,
\(\mathrm{FQRS}^{(s)}_R\) is neither standard- nor average-radius
\[
  \left(
    \frac{1-R}{2}-\gamma,\,
    \ell,\,
    \ell^b-1
  \right)\text{-QLR}.
\]
% In particular,
% \[
%   L^\star_Q\left(\frac{1-R}{2}-\gamma,\ell\right)
%   \ge
%   \ell^{\lceil R_1/\gamma\rceil}.
%   % =
%   % \ell^{\Omega(R_1/\gamma)}.
% \]
\end{theorem}

\begin{proof}
By the construction above, it suffices to verify that an integer \(r\)
satisfying~\eqref{eq:grid-window} exists.  The two lower constraints
require
\(
  r/n\ge (1-\rho)/b
\)
and
\(
  r/n\ge (1-R_1)/(b-1),
\)
which motivates the definition of \(\lambda\).
The first upper constraint is \(r/n\le 1/b\), while the second is
equivalently
\[
  \frac{r}{n}
  \le
  \frac{R_1}{b-1}
  -
  \frac{b}{s(b-1)n},
\]
whose asymptotic upper edge is \(R_1/(b-1)\).

We first verify that the larger of the two lower bounds lies strictly
below both asymptotic upper bounds.  It suffices to check the following
four pairwise inequalities:
\begin{enumerate}[itemsep=2pt,topsep=2pt]
\item
\(
  \frac{1-\rho}{b}<\frac1b
\),
because \(\rho=(1-R)/2-\gamma>0\).

\item
\(
  \frac{1-R_1}{b-1}<\frac{R_1}{b-1}
\),
because \(R_1=(1+R)/2>1/2\).

\item
\(
  \frac{1-R_1}{b-1}<\frac1b
\)
is equivalent to \(bR_1>1\).  Since
\(\gamma<(1-R)/2<R_1\), we have
\(b=\lceil R_1/\gamma\rceil\ge2\); together with
\(R_1>1/2\), this gives \(bR_1>1\).

\item
\(
  \frac{1-\rho}{b}<\frac{R_1}{b-1}.
\)
Since \(1-\rho=R_1+\gamma\), this is equivalent to
\(b<R_1/\gamma+1\), which follows from
\(b=\lceil R_1/\gamma\rceil<R_1/\gamma+1\).
\end{enumerate}
Hence
\(
  \lambda<1/b
\)
and
\(
  \lambda<R_1/(b-1),
\)
so the feasible interval has positive asymptotic width.  In particular,
the denominators appearing below are positive.

Now take \(r:=\lceil\lambda n\rceil\).  Since \(r\ge\lambda n\), the
two lower constraints~\eqref{eq:grid-window}(iii)--(iv) hold
automatically.  Since \(r\le\lambda n+1\),
Condition~\eqref{eq:grid-window}(i), \(br\le n\), follows whenever
\(
  n\ge1/(1/b-\lambda).
\)
Likewise, Condition~\eqref{eq:grid-window}(ii),
\(s(b-1)r+b\le R_1sn\), follows whenever
\(
  n\ge
  (s(b-1)+b)/(R_1s-s(b-1)\lambda).
\)

Since the stated lower bound on \(n\) guarantees all four conditions
in~\eqref{eq:grid-window}, the construction above applies.  
% Therefore
% \[
%   L_Q^\star(\rho,\ell)
%   \ge
%   \ell^b.
%   % =
%   % \ell^{\lceil R_1/\gamma\rceil}.
% \]
\end{proof}

\subsubsection{QLD lower bound: Coset GSB}

\begin{corollary}[FQRS coset Singleton converse]
\label{cor:fqrs-folded-css-qld-lower}
Assume \(R_1n\in\mathbb N\).  For every integer \(L\ge1\), set
\(
  \rho_L
  :=
  \frac1n
  \left\lceil
    \frac{L}{L+1}
    \left(
      \frac{1-R}{2}n+\lceil\log_{q^s}(L+1)\rceil
    \right)
  \right\rceil.
\)
Whenever \(L+1\le(q^s)^{Rn}\), \(Q\) is neither standard- nor
average-radius \((\rho,L)\)-QLD for any \(\rho\ge\rho_L\).
\end{corollary}

\begin{proof}
Apply the conditional folded generalized Singleton bound,
Lemma~\ref{lem:folded-coset-gsb}, with
\(C=C_1\) and \(W=C_2^\perp\).
Recall that
\(
  \dim C_1=R_1sn
\)
and
\(
  \dim C_2^\perp=(1-R_1)sn.
\)
Since \(R_1n\in\mathbb N\), the MDS structure of
\(C_2^\perp\) gives
\(
  \operatorname{ris}_s(C_2^\perp)=(1-R_1)n.
\)
Moreover,
\(
  m_L
  =
  R_1n-\lceil\log_{q^s}(L+1)\rceil.
\)
The condition \(L+1\le(q^s)^{Rn}\) gives
\(
  \lceil\log_{q^s}(L+1)\rceil\le Rn,
\)
and hence
\(
  m_L
  \ge
  R_1n-Rn
  =
  (1-R_1)n
  =
  \operatorname{ris}_s(C_2^\perp).
\)
Thus Lemma~\ref{lem:folded-coset-gsb} applies.   The one-sector lift then gives the
corresponding QLD bad list.
\end{proof}

\subsection{Matching list-size scales}

We now place the achievability and lower bounds side by side.
Write \(R_1:=(1+R)/2\) and
\(\rho:=(1-R)/2-\gamma\).
For either the standard- or average-radius notion, define
\[
  L_Q^\star(\rho,\ell)
  :=
  \min\{\,L:Q\text{ is }(\rho,\ell,L)\text{-QLR}\,\},
\]
and abbreviate
\(
  L_Q^\star(\rho):=L_Q^\star(\rho,1)
\)
for QLD. The bounds below hold under either interpretation, subject
to the corresponding parameter conditions.

\begin{corollary}[FQRS QLR exponent matching]
\label{cor:fqrs-tight}
Fix \(R\in(0,1)\), \(\gamma\in(0,(1-R)/2)\), and
\(2\le\ell\le q\), under the corresponding hypotheses of
Theorems~\ref{thm:fqrs-qlr} and~\ref{thm:fqrs-grid}.
Then
\[
  \ell^{\lceil R_1/\gamma\rceil}
  \le
  L_Q^\star(\rho,\ell)
  \le
  \left\lfloor
    \left(
      \frac{\ell}{R_1+\gamma/2}
    \right)^{3+2R_1/\gamma}
  \right\rfloor.
\]
Consequently,
\[
  \log_\ell L_Q^\star(\rho,\ell)
  =
  \Theta\!\left(\frac{R_1}{\gamma}\right)
  \qquad(\gamma\downarrow0).
\]
Both statements hold in the standard- and average-radius senses
under their respective folding conditions.
\end{corollary}

\begin{proof}
The lower bound follows from Theorem~\ref{thm:fqrs-grid}, while the
upper bound follows from Theorem~\ref{thm:fqrs-qlr}.
Taking logarithms to base \(\ell\) gives the asymptotic conclusion.
\end{proof}

\begin{remark}[Classical exponent and the CSS pairing gain]
\label{rem:fqrs-qlr-classical-exponent}
The QLR upper bound inherits the classical folded-RS list-size
exponent of~\cite{BCDZ26BL}.
A sector-product argument would not change the exponent scaling but
would square the one-sector list size.  The pairing lemma removes this
factor-two loss in the exponent and preserves the one-sector bound for
the joint CSS problem.
\end{remark}

\begin{corollary}[FQRS QLD near-Singleton matching]
\label{cor:fqrs-qld-tight}
Fix \(R\in(0,1)\) and \(\gamma\in(0,(1-R)/2)\), and set
\(
  L:=\lceil(1-R)/\gamma\rceil.
\)
Assume \(R_1n\in\N\) and
\(
  s\ge
  \max\{L,(L-1)(1+2R_1/\gamma)\}.
\)
Then, for all sufficiently large admissible \(n\),
\[
  \left\lceil\frac{1-R}{2\gamma}\right\rceil-1
  \le
  L_Q^\star(\rho)
  \le
  \left\lceil\frac{1-R}{\gamma}\right\rceil.
\]
Consequently,
\[
  L_Q^\star(\rho)
  =
  \Theta\!\left(\frac{1-R}{\gamma}\right)
  \qquad(\gamma\downarrow0).
\]
Both statements hold in the standard- and average-radius senses.
\end{corollary}

\begin{proof}
For the upper bound, recall that \(k_1=R_1ns\) and that
\(L=\lceil(1-R)/\gamma\rceil\).
Theorem~\ref{thm:fqrs-qld} gives average-radius QLD for every radius
below
\(
  \frac{L}{L+1}
  \left(1-\frac{k_1-1}{n(s-L+1)}\right).
\)
Our target radius is \(\rho=1-R_1-\gamma\).
The assumed folding condition gives
\(
  R_1(L-1)/(s-L+1)\le\gamma/2,
\)
and hence
\(
  1-\frac{k_1-1}{n(s-L+1)}
  >
  1-R_1-\frac{\gamma}{2}.
\)
Since
\(
  L=\lceil2(1-R_1)/\gamma\rceil,
\)
we have
\(
  \frac{L}{L+1}(1-R_1-\gamma/2)>\rho.
\)
Thus \(L_Q^\star(\rho)\le L=\lceil(1-R)/\gamma\rceil\).

For the lower bound, consider any integer
\(
  m<\lceil(1-R)/(2\gamma)\rceil-1.
\)
Corollary~\ref{cor:fqrs-folded-css-qld-lower} gives a bad list of
size \(m+1\) at radius
\(
  \frac{m}{m+1}\frac{1-R}{2}+O(1/n).
\)
Since
\(
  \frac{m}{m+1}\frac{1-R}{2}<\rho,
\)
for all sufficiently large admissible \(n\), such a bad list already
exists within radius \(\rho\).  Thus \(m\) is not a valid list budget,
and therefore
\(
  L_Q^\star(\rho)
  \ge
  \lceil(1-R)/(2\gamma)\rceil-1.
\)

The asymptotic conclusion follows from the two-sided bounds.
\end{proof}

\section{Random CSS}
\label{sec:random-css-results}

Throughout this section, \(Q=\CSS(C_1,C_2)\) is sampled from the
balanced ensemble of Definition~\ref{def:random-css}.  As before, write
\(R_1:=(1+R)/2\), so \(1-R_1=(1-R)/2\).  Then
\(C_2^\perp\subseteq C_1\) and
\(\dim C_1=\dim C_2=R_1n\).
Each of \(C_1\) and \(C_2\) is marginally a uniformly random
rate-\(R_1\) linear code, although the two components are not
independent.

\subsection{Achievability: Collision patterns and random-linear certificates}
\label{subsec:random_css_achievability}
\label{subsubsec:random-certificates}

We obtain the required one-sector entropy and remainder certificates
by showing that any certificate violation on a bounded number of
codewords determines a finite coordinate-collision pattern that a
random linear code avoids with high probability.  Once the two
marginal certificates hold, the pairing lemma performs the same
coefficient-preserving joint step as in the FQRS analysis.

Arrange \(b\) candidate codewords as the columns of an \(n\times b\)
matrix,
\[
  M=
  \begin{pmatrix}
    c_1^{(1)} & \cdots & c_1^{(b)}\\
    \vdots & & \vdots\\
    c_n^{(1)} & \cdots & c_n^{(b)}
  \end{pmatrix}.
\]
The \(i\)-th row records their values at coordinate \(i\).
Grouping columns having the same value in this row induces a partition
\(\Pi_i\) of \([b]\), recording the coordinate-\(i\) collision pattern.

Associated with \(\Pi_i\) is the linear subspace
\[
  V_i(\Pi_i)
  :=
  \left\{
    u\in\F_q^b:
    u_j=u_{j'}
    \text{ whenever \(j,j'\) lie in the same block of \(\Pi_i\)}
  \right\}.
\]
The collection
\(
  \mathcal V(\boldsymbol\Pi)
  :=
  (V_1(\Pi_1),\ldots,V_n(\Pi_n))
\)
is called the \emph{local profile} of
\(\boldsymbol\Pi=(\Pi_1,\ldots,\Pi_n)\).
A code \(C\) contains \(\mathcal V(\boldsymbol\Pi)\) if there exist
\(b\) distinct codewords \(c^{(1)},\ldots,c^{(b)}\in C\) such that,
for every \(i\), the \(i\)-th row vector lies in \(V_i(\Pi_i)\).
Note that the profile enforces the equalities within each block of
\(\Pi_i\), while allowing additional equalities between different
blocks.

Let \(p=(p_1,\ldots,p_b)\) be a probability distribution on the
\(b\) candidates, and let \(Y\) take \(c^{(j)}\) with probability
\(p_j\).  For each coordinate \(i\), let the induced block-mass
distribution be
\[
  p^{\Pi_i}(J):=\sum\nolimits_{j\in J}p_j,
  \qquad J\in\Pi_i.
\]
If different blocks of \(\Pi_i\) take the same symbol value, the actual
symbol distribution is obtained from \(p^{\Pi_i}\) by merging blocks.
Such merging does not increase either entropy or remainder.  Hence,
for every coordinate \(i\),
\[
  H(Y_i)\le H(p^{\Pi_i}),
  \qquad
  r(Y_i)\le r(p^{\Pi_i}).
\]

\begin{definition}[Forbidden collision patterns]
\label{def:forbidden-collision-pattern}
Fix \(B\ge2\) and \(0<\delta\le1\).  For \(2\le b\le B\), a collision
pattern \(\boldsymbol\Pi=(\Pi_1,\ldots,\Pi_n)\) is
\emph{entropy-\(\delta\)-forbidden} if there exist positive
probabilities \(p_1,\ldots,p_b\), summing to one, such that
\[
  \sum\nolimits_{i=1}^n H(p^{\Pi_i})
  <\delta n H(p).
\]
It is \emph{remainder-\(\delta\)-forbidden} if there exist such
probabilities satisfying
\[
  \sum\nolimits_{i=1}^n r(p^{\Pi_i})
  <
  \delta n r(p).
\]
\end{definition}

Recall that failure of the bounded-support entropy or remainder
certificate means that there exists a distribution \(Y\), supported on
at most \(B\) codewords, such that
\(
  \sum_i H(Y_i)<\delta n H(Y)
\)
or
\(
  \sum_i r(Y_i)<\delta n r(Y),
\)
respectively.  The next lemma shows that these certificate violations
are equivalent to containing a corresponding forbidden collision
profile.

\begin{lemma}[Certificate failures are collision events]
\label{lem:unified-collision-failure}
Fix \(B\ge2\) and \(0<\delta\le1\).
A code \(C\subseteq\F_q^n\) fails
\(\mathsf{Cert}^{H}_B(\delta)\) if and only if it contains
\(\mathcal V(\boldsymbol\Pi)\) for some
entropy-\(\delta\)-forbidden collision pattern
\(\boldsymbol\Pi\).  Likewise, \(C\) fails
\(\mathsf{Cert}^{r}_B(\delta)\) if and only if it contains
\(\mathcal V(\boldsymbol\Pi)\) for some
remainder-\(\delta\)-forbidden collision pattern.
Moreover, there are at most
\(B\,B^{Bn}\) collision patterns for \(2\le b\le B\).
\end{lemma}

\begin{proof}
Suppose first that \(C\) contains
\(\mathcal V(\boldsymbol\Pi)\) for an
entropy-\(\delta\)-forbidden pattern
\(\boldsymbol\Pi\), witnessed by probabilities
\(p_1,\ldots,p_b\), and let
\(c^{(1)},\ldots,c^{(b)}\) be the witnessing distinct codewords.
Let \(Y\) take \(c^{(j)}\) with probability \(p_j\).
Since the codewords are distinct, \(H(Y)=H(p)\), while
\(H(Y_i)\le H(p^{\Pi_i})\) for every \(i\).  Hence
\[
  \sum_i H(Y_i)
  \le
  \sum_i H(p^{\Pi_i})
  <
  \delta n H(p)
  =
  \delta n H(Y),
\]
so \(\mathsf{Cert}^{H}_B(\delta)\) fails.
The remainder case is identical, using
\(r(Y_i)\le r(p^{\Pi_i})\) and \(r(Y)=r(p)\).

Conversely, suppose \(\mathsf{Cert}^{H}_B(\delta)\) fails.
Then there is a distribution \(Y\), supported on
\(2\le b\le B\) distinct codewords with positive probabilities
\(p_1,\ldots,p_b\), such that
\(
  \sum_i H(Y_i)<\delta n H(Y).
\)
Let \(\Pi_i\) be the actual equality partition at coordinate \(i\).
Then \(Y_i\) has distribution \(p^{\Pi_i}\), so
\(H(Y_i)=H(p^{\Pi_i})\) for every \(i\), and therefore
\[
  \sum_i H(p^{\Pi_i})
  <
  \delta n H(p).
\]
Thus \(\boldsymbol\Pi\) is entropy-\(\delta\)-forbidden, and
\(C\) contains \(\mathcal V(\boldsymbol\Pi)\).
The remainder converse is identical.

Finally, there are at most \(b^b\) partitions of \([b]\), hence at
most \(b^{bn}\) collision patterns for fixed \(b\).  Summing over
\(2\le b\le B\) gives at most \(B\,B^{Bn}\) patterns.
\end{proof}

For an \(s\)-folded code, we use the duplicated-profile convention
of~\cite{BCDZ26Matroid}.  Given
\(
  \mathcal V=(V_1,\ldots,V_n),
\)
define
\[
  \mathcal V^{(s)}
  :=
  (V_1,\ldots,V_1,
   V_2,\ldots,V_2,
   \ldots,
   V_n,\ldots,V_n).
\]
If \(C'\subseteq\F_q^{sn}\) denotes the unfolded scalar code of
\(C\subseteq(\F_q^s)^n\), then \(C\) contains
\(\mathcal V\) precisely when \(C'\) contains
\(\mathcal V^{(s)}\).

\begin{corollary}[Subspace designs avoid forbidden profiles]
\label{cor:design-avoids-forbidden}
Let \(d\ge1\), \(0<\delta\le1\), and let
\(C\subseteq(\F_q^s)^n\) be a \(\mu\)-slacked
\(d\)-subspace-designable code of rate \(R_{\rm c}\).
If \(\delta\le1-R_{\rm c}-\mu\), then \(C\) contains no
entropy-\(\delta\)-forbidden or remainder-\(\delta\)-forbidden profile
on at most \(d+1\) candidates.
\end{corollary}

\begin{proof}
Any set of at most \(d+1\) codewords has an affine message span of
dimension at most \(d\).  Lemma~\ref{lem:design-to-certificate}
therefore gives
\(\mathsf{Cert}^{H}_{d+1}(1-R_{\rm c}-\mu)\) and
\(\mathsf{Cert}^{r}_{d+1}(1-R_{\rm c}-\mu)\), and hence both
certificates with coefficient \(\delta\).
The collision argument of Lemma~\ref{lem:unified-collision-failure}
then gives the claim.
\end{proof}

The following classical transfer result shows that if all sufficiently
folded \(d\)-subspace-designable codes at rate
\(\widehat R_{\rm c}\) avoid a fixed local profile, then a random
scalar linear code at a slightly lower rate contains that profile only
with exponentially small probability.

\begin{corollary}[Fixed-profile transfer;
{\cite[Theorem~4.2 and proof of Corollary~4.3]{BCDZ26Matroid}}]
\label{cor:bcdz-fixed-profile-transfer}
Let \(\mathcal V\) be a \(b\)-local profile.
Suppose that, for all sufficiently large folding parameters, every
\(d\)-subspace-designable folded code of rate
\(\widehat R_{\rm c}\) avoids \(\mathcal V\).
If a random scalar linear code \(C\subseteq\F_q^n\) in the
parity-check model has rate \(R_{\rm c}\) satisfying
\(
  R_{\rm c}
  \le
  \widehat R_{\rm c}-(b^2+1)/n-\varepsilon
\)
for some \(\varepsilon>0\), then
\(
  \Pr[C\text{ contains }\mathcal V]
  \le
  q^{-\varepsilon n+b^2}.
\)
\end{corollary}

We now derive the required bounded-support entropy and remainder
certificates for a random linear code.

\begin{lemma}[Random-linear coordinate certificates]
\label{lem:random-bl-certificate}
Fix \(R_{\rm c}\in(0,1)\), an integer \(B\ge2\), and
\(0<\eta<3(1-R_{\rm c})\).  Assume \(R_{\rm c}n\in\N\),
\(
  n\ge12(B^2+1)/\eta
\)
and
\(
  q\ge B^{24B/\eta}.
\)
Then a uniformly random \(R_{\rm c}n\)-dimensional subspace
\(C\le\F_q^n\) satisfies both
\[
  \mathsf{Cert}^{H}_B
  \left(1-R_{\rm c}-\frac{\eta}{3}\right)
  \qquad\text{and}\qquad
  \mathsf{Cert}^{r}_B
  \left(1-R_{\rm c}-\frac{\eta}{3}\right)
\]
except with probability at most
\(
  Bq^{-\eta n/24+B^2}+o_n(1).
\)
\end{lemma}

\begin{proof}
We first work in the parity-check model, where
\(C=\ker H\) for a uniformly random
\(H\in\F_q^{(1-R_{\rm c})n\times n}\).
Set
\(
  \delta:=1-R_{\rm c}-\eta/3
\)
and
\(
  \widehat R_{\rm c}
  :=R_{\rm c}+\lceil\eta n/6\rceil/n.
\)

Let \(\widehat C\subseteq(\F_q^s)^n\) be any
\((B-1)\)-subspace-designable folded code of rate
\(\widehat R_{\rm c}\).
For every nonzero message subspace \(W\) with \(\dim W\le B-1\),
\[
  \sum_i\dim(W\cap H_i)
  \le
  \widehat R_{\rm c}\,n\dim W+1
  \le
  \left(\widehat R_{\rm c}+\frac1n\right)n\dim W.
\]
Thus \(\widehat C\) is \(1/n\)-slacked
\((B-1)\)-subspace designable.
Moreover,
\(
  \widehat R_{\rm c}
  \le
  R_{\rm c}+\eta/6+1/n.
\)
Since
\(
  n\ge12(B^2+1)/\eta\ge12/\eta,
\)
we have \(2/n\le\eta/6\), and hence
\[
  \widehat R_{\rm c}+\frac1n
  \le
  R_{\rm c}+\frac{\eta}{6}+\frac2n
  \le
  R_{\rm c}+\frac{\eta}{3}
  =
  1-\delta.
\]
In particular,
\(
  \delta\le1-\widehat R_{\rm c}-1/n.
\)
Corollary~\ref{cor:design-avoids-forbidden} therefore shows that
\(\widehat C\) contains no entropy-\(\delta\)-forbidden or
remainder-\(\delta\)-forbidden profile on \(b\le B\) candidates.

Now fix such a forbidden profile on \(b\le B\) candidates.  Since
\[
  \widehat R_{\rm c}-R_{\rm c}-\frac{b^2+1}{n}
  \ge
  \frac{\eta}{6}-\frac{B^2+1}{n}
  \ge
  \frac{\eta}{12},
\]
Corollary~\ref{cor:bcdz-fixed-profile-transfer}, with
\(d=B-1\) and \(\varepsilon=\eta/12\), gives
\[
  \Pr[
    C\text{ contains }\mathcal V(\boldsymbol\Pi)
  ]
  \le
  q^{-\eta n/12+b^2}.
\]

By Lemma~\ref{lem:unified-collision-failure}, failure of either
certificate is represented by one of the forbidden collision profiles
on \(2\le b\le B\) candidates.  Hence
\[
\begin{aligned}
  \Pr[
    \mathsf{Cert}^{H}_B(\delta)
    \text{ or }
    \mathsf{Cert}^{r}_B(\delta)
    \text{ fails}
  ]
  &\le
  \sum_{b=2}^B b^{bn}q^{-\eta n/12+b^2}\\
  &\le
  B\,B^{Bn}q^{-\eta n/12+B^2}\\
  &\le
  Bq^{-\eta n/24+B^2},
\end{aligned}
\]
where the last inequality uses
\(
  q\ge B^{24B/\eta},
\)
which implies
\(
  B^{Bn}\le q^{\eta n/24}.
\)

This proves the bound in the parity-check model.  Conditional on the
parity-check matrix having full row rank, its kernel is uniformly
distributed among the \(R_{\rm c}n\)-dimensional subspaces of
\(\F_q^n\).  By~\cite[Lemma~A.1]{LMS25}, the statistical distance
between the parity-check model and the uniform fixed-dimensional model
is at most
\(
  1-\exp(-nq^{-R_{\rm c}n})=o_n(1).
\)
Thus passing to the uniform fixed-dimensional model adds only
\(o_n(1)\) to the failure probability.  Thus, in the uniform
fixed-dimensional model, the failure probability is at most
\(
  Bq^{-\eta n/24+B^2}+o_n(1),
\)
as claimed.
\end{proof}

\begin{theorem}[Random CSS QLR achievability]
\label{thm:random-css-qlr}
Fix \(R\in(0,1)\), \(\ell\ge2\), and
\(\gamma\in(0,(1-R)/2)\).  Put
\(
  L:=
  \left\lfloor
    \left(
      \frac{\ell}{R_1+\gamma/2}
    \right)^{3+2R_1/\gamma}
  \right\rfloor
\)
and \(B:=L+1\).
Suppose
\(
  q\ge B^{24B/\gamma}
\)
and
\(
  n\ge12(B^2+1)/\gamma.
\)
Then, with probability \(1-o_n(1)\), \(Q\) is standard- and
average-radius
\(
  ((1-R)/2-\gamma,\ell,L)
\)-QLR.
\end{theorem}

\begin{proof}
Since each of \(C_1\) and \(C_2\) is marginally a uniformly random
rate-\(R_1\) linear code, Lemma~\ref{lem:random-bl-certificate}
shows that each satisfies
\(
  \mathsf{Cert}^{H}_B(1-R_1-\gamma/3)
\)
with probability \(1-o_n(1)\).
A union bound over the two failure events therefore shows that both
components satisfy this certificate with probability \(1-o_n(1)\).
The remaining pairing-and-conversion argument is analogous to the
proof of Theorem~\ref{thm:fqrs-qlr}, using the bounded-support pairing
Corollary~\ref{cor:bounded-certificate-pairing}.
\end{proof}

\begin{theorem}[Random CSS QLD achievability]
\label{thm:random-css-qld}
Fix \(R\in(0,1)\), let \(R_1=(1+R)/2\), and fix an integer
\(L\ge1\).  Set \(B:=L+1\).  For every
\(0<\varepsilon<1-R_1\), if
\(
  q\ge B^{8B/\varepsilon}
\)
and
\(
  n\ge4(B^2+1)/\varepsilon,
\)
then, with probability \(1-o_n(1)\), \(Q\) is average-radius
\((\rho,L)\)-QLD for every
\[
  \rho<
  \frac{L}{L+1}(1-R_1-\varepsilon).
\]
\end{theorem}

\begin{proof}
Apply Lemma~\ref{lem:random-bl-certificate} to each component with
\(R_{\rm c}=R_1\) and \(\eta=3\varepsilon\).  A union bound over the two failure events shows that, with
probability \(1-o_n(1)\), both components satisfy
\(
  \mathsf{Cert}^{r}_B(1-R_1-\varepsilon).
\)
The remaining pairing-and-conversion argument is analogous to the
proof of Theorem~\ref{thm:fqrs-qld}, using the bounded-support pairing
Corollary~\ref{cor:bounded-certificate-pairing}.
\end{proof}

\subsection{Quotient-surviving lower bounds}

\subsubsection{QLR lower bound: Dimension count}

We show that, for random CSS codes, a classical bad list
survives the stabilizer quotient with high probability.

\begin{theorem}[Random-CSS QLR lower bound]
\label{thm:random-css-qlr-lower}
Fix \(R\in(0,1)\), \(\gamma\in(0,(1-R)/2)\), and
\(2\le\ell\le q\).  Put
\(
  \rho:=(1-R)/2-\gamma
\)
and
\(
  M:=\ell^{\lfloor R_1/\gamma\rfloor}.
\)
For all sufficiently large admissible \(n\), with probability
\(1-o_n(1)\), \(Q\) is neither standard- nor average-radius
\(
  (\rho,\ell,M-1)
\)-QLR.  
% In particular,
% \[
%   L_Q^\star(\rho,\ell)
%   \ge
%   \ell^{\lfloor R_1/\gamma\rfloor}.
% \]
\end{theorem}

\begin{proof}
Since
\(
  \rho=1-R_1-\gamma,
\)
Theorem~\ref{thm:ls-lr-lower}, with
\(R_{\rm c}=R_1\) and \(\varepsilon=\gamma\), shows that
\(C_1\) is not \((\rho,\ell,M-1)\)-list recoverable.
Hence there are local lists
\(\mathcal S=(S_1,\ldots,S_n)\), with \(|S_i|\le\ell\), and
\(M\) distinct codewords \(x_1,\ldots,x_M\in C_1\) such that
\(
  \dis(x_j,\mathcal S)\le\rho n
\)
for every \(j\).

Now condition on \(C_1\), and fix any such witnessing local lists
\(\mathcal S\) and codewords \(x_1,\ldots,x_M\).
Since
\(C_2^\perp\) is a uniformly random
\((1-R_1)n\)-dimensional subspace of the \(R_1n\)-dimensional
space \(C_1\), for every fixed nonzero \(v\in C_1\),
\[
  \Pr[v\in C_2^\perp\mid C_1]
  =
  \frac{q^{(1-R_1)n}-1}{q^{R_1n}-1}
  \le
  2q^{-Rn}.
\]
There are fewer than \(M^2\) nonzero pairwise differences
\(x_j-x_{j'}\), so
\[
  \Pr\!\left[
    x_j-x_{j'}\in C_2^\perp
    \text{ for some }j\ne j'
    \,\middle|\, C_1
  \right]
  \le
  2M^2q^{-Rn}
  =
  o_n(1).
\]
Hence, with probability \(1-o_n(1)\), the \(M\) candidates are
pairwise distinct modulo \(C_2^\perp\), and therefore give a quantum
bad list of size \(M\) at radius \(\rho\).  Since every candidate
individually has disagreement at most \(\rho n\), the same family is
bad in both the standard- and average-radius senses.
\end{proof}

\subsubsection{QLD lower bound: Coset GSB}

The QLD converse is deterministic.  Every balanced realization
satisfies
\(
  n-\dim C_1=n-\dim C_2=(1-R)n/2,
\)
so the scalar coset generalized Singleton bound applies directly.

\begin{corollary}[Random-CSS coset Singleton converse]
\label{cor:random-css-qld-lower}
For every integer \(L\ge1\), set
\(
  \rho_L
  :=
  \frac1n
  \left\lceil
    \frac{L}{L+1}
    \left(
      \frac{1-R}{2}n+\lceil\log_q(L+1)\rceil
    \right)
  \right\rceil.
\)
Whenever \(L+1\le q^{Rn}\), every realization of the balanced random
CSS ensemble is neither standard- nor average-radius
\((\rho,L)\)-QLD for any \(\rho\ge\rho_L\).
\end{corollary}

\begin{proof}
Every realization is balanced, with
\(
  \dim C_1=\dim C_2=R_1n
\)
and
\(
  \dim C_2^\perp=\dim C_1^\perp=(1-R_1)n.
\)
Thus Lemma~\ref{lem:coset-gsb} applies directly to either quotient
\(C_1/C_2^\perp\) or \(C_2/C_1^\perp\), giving the stated
\(\rho_L\).  The one-sector lift gives the corresponding QLD bad list.
\end{proof}

\subsection{Matching list-size scales}
Throughout this subsection, set
\(\rho:=(1-R)/2-\gamma=1-R_1-\gamma\).
For random CSS code, we use \(L_Q^\star\) with the same standard- or average-radius
interpretation as above.
% For QLD, we
% abbreviate \(L_Q^\star(\rho):=L_Q^\star(\rho,1)\).

% The QLR results above are already parameterized by the target radius,
% whereas the QLD results are naturally stated as fixed-list radius
% tradeoffs.  We now specialize the latter to the near-Singleton radius
% \(\rho\).

\begin{corollary}[Random-CSS QLR exponent matching]
\label{cor:random-css-qlr-tight}
Fix \(R\in(0,1)\), \(\gamma\in(0,(1-R)/2)\), and
\(2\le\ell\le q\), and suppose the field-size hypothesis of
Theorem~\ref{thm:random-css-qlr} holds.
Then, for all sufficiently large admissible \(n\), with probability
\(1-o_n(1)\),
\[
  \ell^{\lfloor R_1/\gamma\rfloor}
  \le
  L_Q^\star(\rho,\ell)
  \le
  \left\lfloor
    \left(
      \frac{\ell}{R_1+\gamma/2}
    \right)^{3+2R_1/\gamma}
  \right\rfloor
\]
in both the standard- and average-radius senses.
Consequently,
\[
  \log_\ell L_Q^\star(\rho,\ell)
  =
  \Theta\!\left(\frac{R_1}{\gamma}\right)
  \qquad(\gamma\downarrow0).
\]
\end{corollary}

\begin{proof}
The upper bound follows from Theorem~\ref{thm:random-css-qlr}, while
the lower bound follows from
Theorem~\ref{thm:random-css-qlr-lower}.
A union bound shows that both hold simultaneously with probability
\(1-o_n(1)\).  The asymptotic conclusion follows from the two-sided
bounds.
\end{proof}

\begin{corollary}[Random-CSS QLD near-Singleton matching]
\label{cor:random-css-qld-tight}
Fix \(R\in(0,1)\) and \(\gamma\in(0,(1-R)/2)\), and set
\(
  L:=\lceil(1-R)/\gamma\rceil,
 ~
  B:=L+1.
\)
Assume
\(
  q\ge B^{16B/\gamma}.
\)
Then, for all sufficiently large admissible \(n\), with probability
\(1-o_n(1)\),
\[
  \left\lceil\frac{1-R}{2\gamma}\right\rceil-1
  \le
  L_Q^\star(\rho)
  \le
  \left\lceil\frac{1-R}{\gamma}\right\rceil
\]
in both the standard- and average-radius senses.
The lower bound holds deterministically for every realization of the
balanced ensemble.  Consequently,
\[
  L_Q^\star(\rho)
  =
  \Theta\!\left(\frac{1-R}{\gamma}\right)
  \qquad(\gamma\downarrow0).
\]
\end{corollary}

\begin{proof}
The argument is analogous to that of
Corollary~\ref{cor:fqrs-qld-tight}.

For the upper bound, apply Theorem~\ref{thm:random-css-qld} with
\(\varepsilon=\gamma/2\).  With \(B=L+1\), its field-size condition becomes
\(
  q\ge B^{16B/\gamma},
\)
while its blocklength condition holds for all sufficiently large
admissible \(n\).  The achievable-radius threshold is therefore
\(
  \frac{L}{L+1}
  \left(1-R_1-\frac{\gamma}{2}\right).
\)
Since
\(
  L=\lceil(1-R)/\gamma\rceil
  =\lceil2(1-R_1)/\gamma\rceil,
\)
the same calculation as in
Corollary~\ref{cor:fqrs-qld-tight} shows that this threshold is
strictly larger than the target radius
\(
  \rho=1-R_1-\gamma.
\)
Hence, with probability \(1-o_n(1)\),
\(
  L_Q^\star(\rho)\le\lceil(1-R)/\gamma\rceil.
\)

For the lower bound, Corollary~\ref{cor:random-css-qld-lower} is the
scalar analogue of the FQRS coset converse.  For any fixed integer
\(m\ge1\), it gives, deterministically for every balanced realization,
a bad list of size \(m+1\) at radius
\(
  \frac{m}{m+1}\frac{1-R}{2}+O(1/n).
\)
Thus the same argument as in
Corollary~\ref{cor:fqrs-qld-tight} excludes every
\(
  m<\lceil(1-R)/(2\gamma)\rceil-1
\)
for all sufficiently large admissible \(n\), and hence
\(
  L_Q^\star(\rho)
  \ge
  \lceil(1-R)/(2\gamma)\rceil-1.
\)

The asymptotic conclusion follows from the two-sided bounds.
\end{proof}

\section*{AI Disclosure}

This project began around January 2026. The overall research direction and the questions addressed in this work were formulated by the authors. Over the course of several months, we used multiple versions of ChatGPT 5 Pro, ChatGPT 6 Astra, and Claude Fable as research assistants through extensive iterative interactions, as well as for editorial assistance.

Substantial human effort was devoted to the development, exposition, organization, and presentation of the results. The authors take full responsibility for the contents of this work and have independently verified all mathematical claims, proofs, and references.

\bibliographystyle{alpha}
\bibliography{refs}

\end{document}